\documentclass[conference]{IEEEtran}
\IEEEoverridecommandlockouts

\usepackage{amsmath,amsthm}
\usepackage{epsfig,psfrag}
\usepackage{graphicx}
\usepackage{graphics}
\usepackage{latexsym}
\usepackage{amssymb}
\usepackage{amsmath}
\usepackage{stmaryrd}
\usepackage{ifthen}
\usepackage{array}
\usepackage{mathrsfs}
\usepackage{paralist}
\usepackage{tikz}
\usetikzlibrary{decorations.pathmorphing}
\usepackage{multirow}
\usepackage{subfig}

\usepackage{pgf}
\usepackage{tikz}
\usetikzlibrary{arrows,automata,backgrounds,decorations}
\usetikzlibrary{shapes.multipart}
\usepgflibrary{decorations.pathreplacing}
\usetikzlibrary{arrows,calc,topaths}
\usetikzlibrary{automata,shapes.arrows, backgrounds, fadings,intersections, positioning}
\usetikzlibrary{shadows}
\tikzstyle{small}=[font=\footnotesize]
\tikzset{
    every picture/.style={>=stealth,auto,node distance=2cm},
}

\usetikzlibrary{petri}
\usetikzlibrary{shadows}
\usetikzlibrary{decorations.pathmorphing}
\usetikzlibrary{decorations.shapes}
\usetikzlibrary{shapes}
\usetikzlibrary{backgrounds}
\usetikzlibrary{fit}
\usetikzlibrary{arrows}
\usetikzlibrary{calc}
\usetikzlibrary{mindmap}
\usetikzlibrary{matrix}

\usepackage[capitalise,noabbrev,nameinlink]{cleveref}

\newcommand{\N}{\mathbb{N}}

\newcommand{\abs}[1]{\lvert#1\rvert}

\newcommand{\problemx}[3]{
\par\noindent\underline{\sc#1}\par\nobreak\vskip.2\baselineskip
\begingroup\clubpenalty10000\widowpenalty10000
\setbox0\hbox{\bf INPUT:\ }\setbox1\hbox{\bf QUESTION:\ }
\dimen0=\wd0\ifnum\wd1>\dimen0\dimen0=\wd1\fi
\vskip-\parskip\noindent
\hbox to\dimen0{\box0\hfil}\hangindent\dimen0\hangafter1\ignorespaces#2\par
\vskip-\parskip\noindent
\hbox to\dimen0{\box1\hfil}\hangindent\dimen0\hangafter1\ignorespaces#3\par
\endgroup}

\newcommand{\dist}{\mathcal{D}}
\newcommand{\supp}{{\sf supp}}
\newcommand{\cParity}[1]{{#1}\text{-}\mathtt{Parity}}
\newcommand{\Parity}[1]{\mathtt{Parity}(#1)}
\newcommand{\reach}[1]{\mathtt{Reach}(#1)}

\newcommand{\safety}[1]{\mathtt{Safety}(#1)}
\newcommand{\always}{{\sf G}}
\newcommand{\eventually}{{\sf F}}
\renewcommand{\next}{{\sf X}}

\newtheorem{theorem}{Theorem}
\newtheorem{lemma}[theorem]{Lemma}

\newtheorem{proposition}[theorem]{Proposition}
\newtheorem{definition}{Definition}
\newtheorem{remark}{Remark}

\newenvironment{qtheorem}[1]{%
{\medskip\noindent\bfseries Theorem #1.\hspace{0mm}}
\begin{itshape}%
}{%
\end{itshape}%
}
\newenvironment{qlemma}[1]{%
{\medskip\noindent\bfseries Lemma #1.\hspace{0mm}}
\begin{itshape}%
}{%
\end{itshape}%
}

\newenvironment{qproposition}[1]{%
{\medskip\noindent\bfseries Proposition #1.\hspace{0mm}}%
\begin{itshape}%
}{%
\end{itshape}%
}
\newenvironment{qremark}[1]{%
{\medskip\noindent\bfseries Remark #1.\hspace{0mm}}
\begin{itshape}%
}{%
\end{itshape}%
}

\newcommand{\hide}[1]{}

\newcommand{\lrc}[1]{(#1)}
\newcommand{\lrd}[1]{\{#1\}}

\newcommand{\ignore}[1]{}

\newcommand{\emptyword}{\varepsilon}

\newcommand{\nat}{\mathbb N}

\newcommand{\setcomp}[2]{\lrd{{#1}\mid{#2}}}
\newcommand{\tuple}[1]{\lrc{#1}}
\newcommand{\bigdenotationof}[2]{\Big\llbracket #1\Big\rrbracket^{#2}}
\newcommand{\denotationof}[2]{\llbracket #1\rrbracket^{#2}}
\newcommand{\mdp}{{\mathcal M}}

\newcommand{\mdptuple}{\tuple{\states,\zstates,\rstates,\transition,\probp}}
\newcommand{\mdpclass}{{\mathcal C}}

\newcommand{\states}{S}

\newcommand{\stateset}{Q}
\newcommand{\state}{s}

\newcommand{\zstates}{\states_\zsymbol}

\newcommand{\rstates}{\states_\rsymbol}

\newcommand{\zsymbol}{\Box}

\newcommand{\rsymbol}{{\scriptscriptstyle\bigcirc}}

\newcommand{\transition}{{\longrightarrow}}

\newcommand{\probp}{P}

\newcommand{\play}{w}
\newcommand{\playset}{{\mathfrak R}}

\newcommand{\partialplay}{w}

\newcommand{\zstrat}{\sigma}

\newcommand{\zstratset}{\Sigma}

\newcommand{\memory}{{\sf M}}
\newcommand{\memconf}{{\sf m}}
\newcommand{\memsuc}{\pi_s}
\newcommand{\memup}{\pi_u}
\newcommand{\memstrattuple}{\tuple{\memory,\initmem,\memup,\memsuc}}
\newcommand{\initmem}{\memconf_0}
\newcommand{\memstrat}[1]{{\sf T}^{#1}}
\newcommand{\memstratn}{\memstrat{}}

\newcommand{\probm}{{\mathcal P}}

\newcommand{\formula}{{\varphi}}

\newcommand{\colorset}[3]{[#1]^{\coloring#2#3}}
\newcommand{\coloring}{{\mathit{C}ol}}
\newcommand{\colorof}[1]{\coloring\lrc{{#1}}}

\newcommand{\zwinset}[2]{\big [#1 \big ]^{{#2}}}

\newcommand{\reachset}{T}%{{\mathcal{T}\,\,\!\!}}

\newcommand{\cset}{{\mathcal C}}

\newcommand{\valueof}[2]{{\mathtt{val}_{#1}(#2)}}

\newcommand{\constraint}{\rhd}

\newcommand{\const}{c}

\newcommand{\safe}[1]{{\it Safe}(#1)}
\newcommand{\safesub}[2]{{\it Safe_{#1}}(#2)}
\mathchardef\mhyphen="2D % define a "math hyphen"
\newcommand{\optav}{\zstrat_{\mathit{opt\mhyphen av}}}

\begin{document}
\pagestyle{plain}
\title{Parity Objectives in Countable MDPs}

\author{
\IEEEauthorblockN{Stefan Kiefer\IEEEauthorrefmark{1},
Richard Mayr\IEEEauthorrefmark{2},
Mahsa Shirmohammadi\IEEEauthorrefmark{1},
Dominik Wojtczak\IEEEauthorrefmark{3}}
\IEEEauthorblockA{\IEEEauthorrefmark{1}University of Oxford, UK}
\IEEEauthorblockA{\IEEEauthorrefmark{2}University of Edinburgh, UK}
\IEEEauthorblockA{\IEEEauthorrefmark{3}University of Liverpool, UK}
}

\maketitle

\begin{abstract}
We study countably infinite MDPs with parity objectives, 
and special cases with a bounded number of colors in the Mostowski hierarchy
(including reachability, safety, B\"uchi and co-B\"uchi).

In finite MDPs there always exist optimal 
memoryless deterministic (MD) strategies 
for parity objectives,
but this does not generally hold for countably infinite MDPs.
In particular, optimal strategies need not exist.

For countable infinite MDPs, we provide a 
complete picture of the memory requirements of 
optimal (resp., $\epsilon$-optimal) strategies for all objectives 
in the Mostowski hierarchy.

In particular, there is a strong dichotomy between two different types of
objectives. For the first type, optimal strategies, if they exist, can be
chosen MD, while for the second type optimal strategies require infinite
memory.
%(I.e., there is no third type for which finite-memory or positional
%randomized strategies are both required and sufficient.)
(I.e., for all objectives in the Mostowski hierarchy, if finite-memory randomized strategies suffice then also MD-strategies suffice.)
Similarly, some objectives admit $\epsilon$-optimal MD-strategies,
while for others $\epsilon$-optimal strategies require infinite memory.
Such a dichotomy also holds for the subclass of countably infinite MDPs that are 
finitely branching, though more objectives admit MD-strategies here.
\end{abstract}


\section{Introduction}\label{sec:introduction}

Markov decision processes (MDPs) are a standard model for dynamic systems that
exhibit both stochastic and controlled behavior \cite{Puterman:book}.
The system starts in the initial state and makes a sequence of transitions between states. Depending on the type of the current state, either the controller gets to
choose an enabled transition (or a distribution over transitions), or the next
transition is chosen randomly according to a defined distribution.
By fixing a strategy for the controller, one obtains a probability space
of plays of the MDP.
The goal of the controller is to optimize the expected value of
some objective function on the plays of the MDP.
The fundamental questions are ``what is the optimal value that the controller
can achieve?'', ``does there exist an optimal strategy, or only
$\epsilon$-optimal approximations?'', and ``which types of strategies are
optimal or $\epsilon$-optimal?''.

Such questions have been studied extensively for finite MDPs (see
e.g.~\cite{chatterjee2012survey} for a survey) and also for certain types of
countably infinite MDPs \cite{Puterman:book,Ornstein:AMS1969}.
However, the literature on countable MDPs is mainly focused on objective functions defined w.r.t.~numeric
costs (or rewards) that are assigned to transitions, e.g.~(discounted) expected total reward
or limit-average reward. In contrast, we study qualitative objectives that are
expressed by Parity conditions and which are motived by formal verification questions.

There are works that studied particular classes
of countably infinite, but finitely branching, 
MDPs that arise from models in automata theory
\cite{Etessami:Yannakakis:ICALP05,BBS:ACM2007,brazdil2008reachability,BBEKW10,ACMSS:FOSSACS2016}.
In each of these papers, a crucial part of the analysis is establishing the existence of optimal strategies of particular structure and memory requirements, but none of them looked at proving
such properties for general countable MDPs.
Countable MDPs also naturally occur in the analysis of queueing systems \cite{kitaev1995controlled}, gambling \cite{BergerKSV08}, and branching processes \cite{pliska1976optimization}, which have multiple applications.
They also show up in the analysis of finite-state models, e.g.~in two-player stochastic games \cite{shapley1953stochastic,condon1992complexity} when reasoning about an op\-ti\-mal strategy against a fixed (randomised and memory-full) strategy of the opponent. 

\begin{figure*}[t]
    \centering
    \begin{minipage}{.35\textwidth}
        \begin{center}				
        \scalebox{.9}{\centering
\tikzstyle{dashdotted}=[dash pattern=on 3pt off 2pt on \the\pgflinewidth off 2pt]
\tikzstyle{densely dashdotted}=      [dash pattern=on 3pt off 1pt on \the\pgflinewidth off 1pt]
\tikzstyle{loosely dashdotted}=      [dash pattern=on 3pt off 4pt on \the\pgflinewidth off 4pt]
\begin{tikzpicture}[>=latex',shorten >=1pt,node distance=1.9cm,on grid,auto]

\node (safe)  at(.8,-1) [draw=none] {$\mathtt{Safety}$};
\node (reach) at(4.2,-1)[draw=none] {$\mathtt{Reach}$};
\node (s01)   at(.8,.5) [draw=none] {$\cParity{\{0,1\}}$};
\node (s12)   at(4.2,.5)[draw=none] {$\cParity{\{1,2\}}$};
\node (s012)  at(.8,2) [draw=none] {$\cParity{\{0,1,2\}}$};
\node (s123)  at(4.2,2) [draw=none] {$\cParity{\{1,2,3\}}$};
\node (s012)  at(.8,3.5) [draw=none] {$\cParity{\{0,1,2,3\}}$};
\node (s123)  at(4.2,3.5) [draw=none] {$\cParity{\{1,2,3,4\}}$};

\draw[thick,blue] (2.5,-1.5) --(2.5,-.7)  -- (5.5,.1)--(5.5,-1.5)--(2.5,-1.5);
%\draw[loosely dashdotted,blue] (-.5,5.5)  -- (-.5,-1.5)--(2.5,-1.5);
%\draw[loosely dashdotted,blue] (5.5,0)--(5.5,5.5);

\draw[thick,magenta] (-.7,-.5)  -- (5.7,1.45)--(5.7,-1.7)--(-.7,-1.7)--(-.7,-.5);
\draw[dashed] (-.7,4)  -- (-.7,-.5);
\draw[dashed](5.7,1.45)--(5.7,4);

\node[label=below:\rotatebox{15}{\textbf{\textcolor{blue}{$\epsilon$-optimal}}}] at (3.3,-.2) {};
\node[label=below:\rotatebox{15}{\textbf{\textcolor{magenta}{optimal}}}] at (1.2,.3) {};
%\node[label] at (1.3,-2) {\textbf{a)}~Infinitely branching MDPs};

\end{tikzpicture}}
				\end{center}
				~~\textbf{a)}~Infinitely branching MDPs
    \end{minipage}%
		\quad \quad \quad \quad
    \begin{minipage}{0.35\textwidth}
        \begin{center}
        \scalebox{.9}{\centering
\tikzstyle{dashdotted}=[dash pattern=on 3pt off 2pt on \the\pgflinewidth off 2pt]
\tikzstyle{densely dashdotted}=      [dash pattern=on 3pt off 1pt on \the\pgflinewidth off 1pt]
\tikzstyle{loosely dashdotted}=      [dash pattern=on 3pt off 4pt on \the\pgflinewidth off 4pt]
\begin{tikzpicture}[>=latex',shorten >=1pt,node distance=1.9cm,on grid,auto]

\node (safe)  at(.8,-1) [draw=none] {$\mathtt{Safety}$};
\node (reach) at(4.2,-1)[draw=none] {$\mathtt{Reach}$};
\node (s01)   at(.8,.5) [draw=none] {$\cParity{\{0,1\}}$};
\node (s12)   at(4.2,.5)[draw=none] {$\cParity{\{1,2\}}$};
\node (s012)  at(.8,2) [draw=none] {$\cParity{\{0,1,2\}}$};
\node (s123)  at(4.2,2) [draw=none] {$\cParity{\{1,2,3\}}$};
\node (s012)  at(.8,3.5) [draw=none] {$\cParity{\{0,1,2,3\}}$};
\node (s123)  at(4.2,3.5) [draw=none] {$\cParity{\{1,2,3,4\}}$};

\draw[thick,blue] (-.5,1.35)  -- (5.5,-.4)--(5.5,-1.5)--(-.5,-1.5)--(-.5,1.35);
%\draw[loosely dashdotted,blue] (-.5,5.5)  -- (-.5,1.8);
%\draw[loosely dashdotted,blue](5.5,.2)--(5.5,5.5);

\draw[thick,magenta] (-.7,3.1)  -- (5.7,0.9)--(5.7,-1.7)--(-.7,-1.7)--(-.7,3.1);
\draw[dashed] (-.7,4)  -- (-.7,3.1);
\draw[dashed] (5.7,.9)--(5.7,4);

\node[label=below:\rotatebox{-15}{\textbf{\textcolor{blue}{$\epsilon$-optimal}}}] at (3.7,0.3) {};
\node[label=below:\rotatebox{-15}{\textbf{\textcolor{magenta}{optimal}}}] at (3.2,1.9) {};
%\node[label] at (1.3,-2) {\textbf{b)}~Finitely branching MDPs};

\end{tikzpicture}}
				\end{center}
				~~\textbf{b)}~Finitely branching MDPs
    \end{minipage}
		\caption{For countable MDPs, these diagrams show the memory requirements of optimal and $\epsilon$-optimal 
		strategies for  objectives in the Mostowski hierarchy. 
		An objective in a level of the hierarchy subsumes all objectives in lower levels, e.g.,
		$\cParity{\{0,1,2\}}$ subsumes~$\cParity{\{1,2\}}$. 
		We have extended the Mostowski hierarchy to include
                reachability and safety. 
    % (see Section~\ref{sec-preliminaries}). 
		The magenta (resp., blue) regions enclose objectives where memoryless deterministic (MD) strategies are sufficient for
		optimal (resp., $\epsilon$-optimal) strategies; for objectives
                outside the regions, infinite-memory
		strategies are necessary. 	
		The left  diagram is for infinitely branching MDPs; 
		e.g., $\epsilon$-optimal strategies for all but reachability objectives 
		 require infinite memory, 
		whereas MD-strategies are sufficient for reachability. 
		The right  diagram is for finitely branching MDPs; 
		e.g., optimal strategies (if they exist) can be chosen MD 
                for all objectives subsumed by~$\cParity{\{0,1,2\}}$.
		}
\label{fig:results}
\end{figure*}
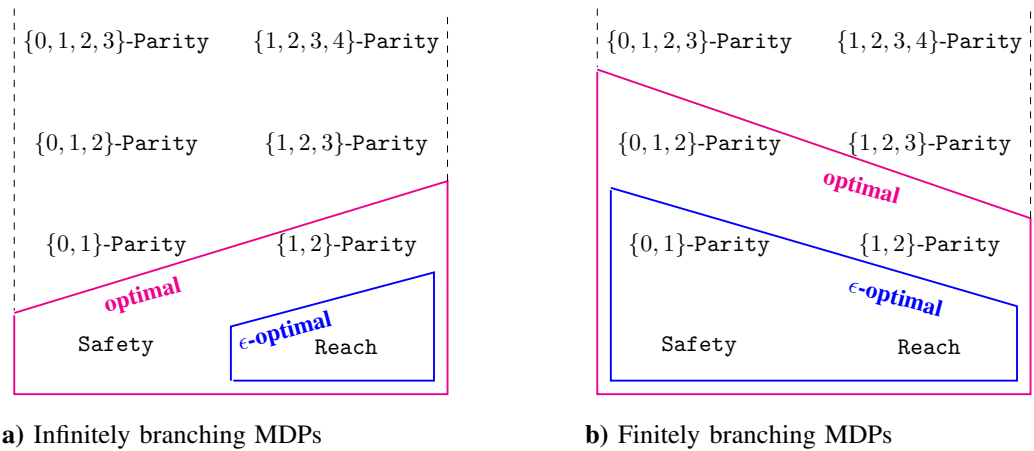

\medskip
{\bf\noindent Finite MDPs vs. Infinite MDPs:}
It should be noted that many standard properties (and proof techniques)
of finite MDPs do \emph{not}
carry over to infinite MDPs. 

E.g., given some objective, consider the set of all states in an MDP
that have nonzero value.
If the MDP is finite then this set is finite and thus there exists some
minimal nonzero value.
This property does \emph{not carry over} to infinite MDPs. 
Here the set of states is infinite and the infimum over the nonzero values can be zero.
As a consequence, even for a reachability objective, it is possible that all states have 
value $>0$, but still the value of some states is $<1$. Such phenomena appear
already in infinite-state Markov chains like the classic Gambler's ruin problem with unfair
coin tosses in the player's favor (0.6 win, 0.4 lose). The value,
i.e., the probability of ruin, is always $>0$, but still $<1$ in every state
except the ruin state itself; cf.~\cite{Feller:book} (Chapt.~14). 
Another difference is that optimal strategies need not exist, even for 
qualitative objectives like reachability or parity. Even if some state
has value $1$, there might not be any single strategy that attains the value $1$, but 
only an infinite family of $\epsilon$-optimal strategies for every $\epsilon >0$.

\medskip
{\bf\noindent Parity objectives:}
We study general countably infinite MDPs with parity objectives.
Parity conditions are widely used in temporal
logic and formal verification, e.g., they can express $\omega$-regular
languages and modal $\mu$-calculus \cite{GTW:2002}.
Every state has a \emph{color}, out of a finite set of colors encoded as natural
numbers. An infinite play is winning iff the highest color that is seen
infinitely often in the play is even. The controller wants to maximize 
the probability of winning plays.
Subclasses of parity objectives are defined by restricting the set of used
colors; these are classified in the Mostowski hierarchy \cite{Mostowski:84}
which includes, e.g., B\"uchi and co-B\"uchi objectives.
Such prefix-independent infinitary objectives cannot generally be encoded by numeric
transition rewards as in \cite{Puterman:book}, though both types subsume 
the simpler reachability and safety objectives.

There are different types of strategies, depending on whether one can take
the whole history of the play into account (history-dependent; (H)),
or whether one is limited to a finite amount of memory (finite memory; (F))
or whether decisions are based only on the current state (memoryless; (M)).
Moreover, the strategy type depends on whether the controller can
randomize (R) or is limited to deterministic choices (D).
The simplest type MD refers to memoryless deterministic strategies.

The type of strategy needed for an optimal (resp.\ $\epsilon$-optimal)
strategy for some objective is also called the \emph{strategy complexity} of the objective.
For finite MDPs, MD-strategies are sufficient for all types of qualitative and
quantitative parity objectives \cite{CAH:QEST2004,Chatterjee:2004:QSP:982792.982808}, but the picture is
more complex for countably infinite MDPs.

Since optimal strategies need not exist in general, we consider both the
strategy complexity of $\epsilon$-optimal strategies, and 
the strategy complexity of optimal strategies under the assumption that they
exist.
E.g., if an optimal strategy exists, can it be chosen MD?

We provide a complete picture of the memory requirements for objectives in the
Mostowski hierarchy, which is summarized in Figure~\ref{fig:results}.

In particular, our results show that there is a strong dichotomy between two different classes of
objectives. For objectives of the first class, optimal strategies, 
where they exist, can be chosen MD.
For objectives of the second class, optimal strategies require infinite memory
in general, in the sense that all FR-strategies achieve the objective
only with probability zero.
A similar dichotomy applies to $\epsilon$-optimal strategies.
For certain objectives, $\epsilon$-optimal MD-strategies exist, while
for all others even $\epsilon$-optimal strategies require infinite memory
in general.
This is a strong dichotomy because there are no objectives in the Mostowski
hierarchy for which other types of strategies (MR, FD, or FR) are both
necessary and sufficient.
Put differently, for all objectives in the Mostowski hierarchy, if FR-strategies suffice then MD-strategies suffice as well.

We also consider the subclass of countable MDPs that are finitely branching.
(Note that these generally still have an infinite number of states.)
The above mentioned dichotomies apply here as well, though the classes of
objectives where optimal (resp.\ $\epsilon$-optimal) strategies can be chosen MD
are larger than for general countable MDPs.

\medskip
{\bf\noindent Outline of the results:}
In Section~\ref{sec-preliminaries} we define countably infinite MDPs, strategies and parity
objectives.
In Section~\ref{sec-parity} we show examples that demonstrate that certain objectives
require infinite memory. For some of these we refer to previous work.
The main new result in this section is Theorem~\ref{thm:parity123} that shows 
that even almost-sure $\cParity{\{1,2,3\}}$ on finitely branching MDPs requires
infinite memory.
These negative results highlight the questions which other objectives still
allow MD-strategies. Apart from the case of reachability objectives, 
these questions were open. We provide complete answers in several steps.
First, in Section~\ref{sec-almost-to-optimal}, we prove a general result
(Theorem~\ref{thm:reduction-to-as})
that relates the strategy complexity of almost-sure winning strategies and
optimal strategies. The complexity of the proof is due to the fact that
we consider \emph{infinite} MDPs (which do not satisfy basic properties of 
finite MDPs in general; see above). We then use this theorem to establish
MD-strategies for B\"uchi, co-B\"uchi and $\cParity{\{0,1,2\}}$ objectives in
the following sections.
In Section~\ref{sec-inf-MDPs} we show that optimal strategies for
B\"uchi objectives, where they exist, can be chosen MD, 
even for infinitely branching MDPs.
In Section~\ref{sec-finite-branch} we consider finitely branching MDPs.
We show that optimal strategies for $\cParity{\{0,1,2\}}$, where they exist,
can be chosen MD (Theorem~\ref{thm:012quant}). This is a very general result. E.g., this question
had been open (and is non-trivial) even for almost-sure co-B\"uchi objectives.
Moreover, we show that $\epsilon$-optimal strategies for co-B\"uchi objectives
can be chosen MD (Theorem~\ref{thm:coBuchi}).
We conclude the paper with a discussion of how some results change when 
one considers uncountable MDPs.

\section{Preliminaries}\label{sec-preliminaries}

%Let $\R$ and $\N$ denote the set of real and  natural numbers, respectively.
%
A \emph{probability distribution} over a countable (not necessarily finite) set~$S$ is a function
$f:\states\to[0,1]$ s.t.~$\sum_{\state\in \states}f(\state)=1$.
We use~$\supp(f) = \{\state \in \states \mid f(\state) > 0\}$ to denote the \emph{support} of~$f$.
Let $\dist(\states)$ be the set of all probability distributions over~$\states$.

We consider countably infinite Markov decision processes (MDPs)~$\mdp=\mdptuple$ where
the countable set~$\states$ of \emph{states}  is partitioned into
the  set~$\zstates$ of states  of the player  and
\emph{random states} $\rstates$.
The relation $\mathord{\transition}\subseteq\states\times\states$ is the transition relation.
We write $\state\transition{}\state'$ if $\tuple{\state,\state'}\in \transition$, and we assume that each state~$\state$ has a \emph{successor} state $\state'$ with $\state \transition \state'$.
The probability function $\probp:\rstates \to \dist(\states)$ assigns to each random 
state~$\state \in \rstates$ a probability distribution over its 
successor states. %~$\{\state' \in \states\mid \state\transition{}\state'\}$.
A set $\reachset \subseteq \states$ is a \emph{sink} in~$\mdp$ if for all $\state \in \reachset$ all successors of~$\state$ are in~$\reachset$.
%we have $\{\state' \in \states\mid \state\transition{}\state'\} \subseteq \reachset$.
The MDP~$\mdp$ is called \emph{finitely branching} if each state has only finitely many successors;
otherwise, it is \emph{infinitely branching}.
A Markov chain is an MDP where~$\zstates=\emptyset$, i.e., all states are random states.

We describe the behavior of an MDP as a one-player stochastic game played
for infinitely many rounds. The game
starts  in a given  initial state~$\state_0$. In each round,  if the game is in state $\state \in \zstates$ 
then the player (or controller) chooses a successor state~$s'$ with~$s\transition{}s'$;
otherwise the game is in a random state~$s\in \rstates$ and proceeds randomly to~$s'$ with probability~$\probp(s)(s')$.

\medskip
\noindent{\bf Strategies.}
A \emph{play}~$\play$  is an infinite sequence
$\state_0\state_1\cdots$ of states 
such that  $\state_i\transition{}\state_{i+1}$ for all $i\geq 0$;
let  $\play(i)=\state_i$ denote the $i$-th state along~$\play$.
A \emph{partial play} is a finite prefix of a play.
We say that (partial) play $\play$ \emph{visits} $\state$ if
$\state=\play(i)$ for some $i$, and that~$\play$ starts in~$s$ if $\state=\play(0)$. 
A \emph{strategy} %of the player
is a function $\zstrat:\states^*\zstates \to \dist(S)$ that assigns to partial plays 
$\partialplay\state \in \states^*\zstates$ 
a distribution over the successors~$\{\state'\in \states\mid \state \transition{} \state'\}$. 
The set of all strategies  in 
$\mdp$ is denoted by $\zstratset_\mdp$ (we omit the subscript and write~$\zstratset$ if $\mdp$ is clear).
A (partial) play~$\state_0\state_1\cdots$ is induced by strategy~$\zstrat$
if~$\state_{i+1}\in \supp(\zstrat(\state_0\state_1\cdots\state_i))$ for all~$i$ with~$\state_i \in \zstates$, and
$\state_{i+1}\in \supp(\probp(\state_i))$ for all~$i$ with ~$\state_i \in \rstates$.
%
%We use~$\playsof{\mdp,\state,\zstrat}$ to denote the set of all plays in $\mdp$ starting in state $\state$ induced by $\zstrat$.

Since this paper focuses on the memory requirements of strategies,
we present an equivalent formulation of strategies, emphasizing 
the amount of memory required to implement a strategy.
Strategies  can
be implemented by probabilistic transducers~$\memstratn=\memstrattuple$
where $\memory$ is a  countable set (the memory of the strategy),
$\memconf_0\in\memory$ is the \emph{initial memory mode}
and  $\states$ is the input and output alphabet.
The probabilistic transition function~$\memup : \memory \times \states\to \dist(\memory)$ updates the 
memory mode of transducer. 
The probabilistic successor function~$\memsuc : \memory \times \zstates \to \dist(\states)$ 
outputs the next successor, where $\state' \in \supp(\memsuc(\memconf,\state))$ implies 
$\state\transition{}\state'$.
We extend $\memup$ to $\dist(\memory) \times \states\to \dist(\memory)$ and  
$\memsuc$ to $\dist(\memory) \times \zstates\to \dist(\states)$, in the natural way.
Moreover, we extend $\memup$
to paths  by
$\memup(\memconf,\emptyword)=\memconf$ and
$\memup(\memconf,\state_0\cdots\state_n) =
\memup(\memup(\state_0\cdots\state_{n-1},\memconf),\state_n)$.
The strategy $\zstrat_{\memstratn}:\states^*\zstates\to\dist(\states)$
induced by the transducer~$\memstratn$ is given by 
$\zstrat_{\memstratn}(\state_0\cdots\state_{n}):=\memsuc(\state_{n},\memup(\state_0\cdots\state_{n-1},\initmem))$.
Note that such strategies allow for randomized memory updates and probabilistic successor functions.

%\TODO{ Dominik suggested to add a sentence to compare two definitions of finite-memory strar..}

Strategies are in general \emph{history dependent}~(H) and \emph{randomized}~(R).
An H-strategy~$\zstrat$ is \emph{finite memory}~(F) if there exists some transducer $\memstratn$ with memory~$\memory$
such that  $\zstrat_{\memstratn}=\zstrat$ and $\abs{\memory}<\infty$; otherwise we say $\zstrat$ \emph{requires infinite memory}.
An F-strategy is   \emph{memoryless}~(M) (also called \emph{positional}) 
if~$\abs{\memory}=1$. 
We may view 
%deterministic strategies as functions $\zstrat: \states^{*} \zstates\to \states$, and 
M-strategies as functions $\zstrat: \zstates \to \dist(\states)$.
An R-strategy $\zstrat$ is \emph{deterministic}~(D)
if $\memup$ and $\memsuc$ map to Dirac distributions;
it implies that $\zstrat(\partialplay)$ is a Dirac distribution for all partial plays~$\partialplay$.
All combinations of the properties in $\{\text{M},\text{F},\text{H}\}\times \{\text{D},\text{R}\}$ are possible, e.g., MD stands for
memoryless deterministic. HR strategies are the most general type.
%For $y \in \{\text{M},\text{F},\text{H}\}$ and $z \in \{\text{D},\text{R}\}$ we let
%$\zallstrats{yz}(\mdp)$  denote
%the set of all strategies of the player in $\mdp$ that satisfy properties $y$ and~$z$.
 
\medskip
\noindent {\bf Probability Measures.}
An MDP $\mdp=\mdptuple$, an initial state~$\state_0$, and a strategy~$\zstrat$ induce a standard probability measure on sets of infinite plays.
We write $\probm_{\mdp,\state_0,\zstrat}({\playset})$ for the probability of a measurable set $\playset \subseteq \state_0 \states^\omega$ of plays starting from~$\state_0$.
It is defined, as usual, by first defining it on the \emph{cylinders} $s_0 s_1 \ldots s_n \states^\omega$, where $s_1, \ldots, s_n \in \states$:
if $s_0 s_1 \ldots s_n$ is not a partial play induced by~$\zstrat$ then set $\probm_{\mdp,\state_0,\zstrat}(s_0 s_1 \ldots s_n \states^\omega) = 0$; 
otherwise set $\probm_{\mdp,\state_0,\zstrat}(s_0 s_1 \ldots s_n \states^\omega) = \prod_{i=0}^{n-1} \bar{\zstrat}(s_0 s_1 \ldots s_i)(s_{i+1})$, where $\bar{\zstrat}$ is the map that extends~$\zstrat$ by $\bar{\zstrat}(w s) = \probp(s)$ for any $w s \in \states^* \rstates$.
Using Carath\'eodory's extension theorem~\cite{billingsley-1995-probability}, this defines a unique probability measure~$\probm_{\mdp,\state_0,\zstrat}$ on measurable subsets of~$s_0 \states^\omega$.

\medskip
\noindent {\bf Objectives.} 
Let $\mdp = \mdptuple$ be an MDP.
The objective of the player is determined by a predicate on
infinite plays.
We assume familiarity with the syntax and semantics of the temporal
logic LTL \cite{CGP:book}.
% (see, e.g., \cite{CGP:book}).
%
Formulas are interpreted on the structure $(\states,\transition)$.
We use 
$\denotationof{\formula}{\state} \subseteq \state \states^\omega$ to denote the set of plays starting from
$\state$ that satisfy the LTL formula $\formula$.
This set is measurable~\cite{Vardi:probabilistic},
and we just write
$\probm_{\mdp,\state,\zstrat}(\formula)$
instead of 
$\probm_{\mdp,\state,\zstrat}(\denotationof\formula\state)$.
We also write $\denotationof{\formula}{}$ for $\bigcup_{s \in S} \denotationof{\formula}{s}$.
%For $\constraint\in\{\geq,>\}$ we write 
%$\probm_{\mdp,\state,\zstrat}(\formula) \constraint \const$
%to compare the probability with a constant $\const$.

Given a target set $\reachset \subseteq \states$, the \emph{reachability objective} is defined by $\reach{\reachset}=\denotationof{\eventually \,\reachset}{}$, i.e., $\state_0\state_1\cdots\in \reach{\reachset} \,\Leftrightarrow\, \exists i.\,
\state_i \in \reachset$.
%Moreover, $\reachn{n}{\reachset}$ denotes the set of all 
%plays visiting~$\reachset$ in the first $n$ steps, i.e.,
%$\state_0\state_1\cdots\in \reachn{n}{\reachset} \,\Leftrightarrow\, \exists
%i\le n.\, \state_i \in \reachset$.
The \emph{safety objective} is defined by $\safety{\reachset}=\denotationof{\always \,\neg \reachset}{}$, i.e., $\state_0\state_1\cdots\in \safety{\reachset} \,\Leftrightarrow\, \forall i.\,
\state_i \not\in \reachset$.
Given a reachability or a safety objective, we can assume without loss of generality that $\reachset$ is a sink in~$\mdp$.

Let $\cset \subseteq \nat$ be a finite set of colors.
A \emph{color function} $\coloring:\states\to \cset$ assigns to each state~$\state$ its color~$\colorof\state$.
For $n\in\nat$, $\mathord{\constraint}\in\{\mathord{<},\mathord{\le},\mathord{=},\mathord{\geq},\mathord{>}\}$ and $\stateset\subseteq\states$, let $\colorset
\stateset\constraint n:=\setcomp{\state\in \stateset}{\colorof\state\constraint n}$
be the set of states in $\stateset$ with color $\constraint n$.
The \emph{parity objective} is defined by
\[
 \Parity{\coloring} := \bigdenotationof{\bigvee_{i\in \cset}\left(\always\eventually \colorset{\states}{=2\cdot i}{} \wedge
 \eventually\always \colorset{\states}{\leq 2\cdot i}{}\right)
}{},
\]
i.e., $\Parity{\coloring}$ is the set of infinite plays such that the largest color that occurs infinitely often along the play is even.

The Mostowski hierarchy \cite{Mostowski:84} classifies parity objectives
by restricting the range of the function~$\coloring$ to a set of colors $\cset \subseteq \nat$.
We write $\cParity{\cset}$ for such restricted parity objectives.
In particular,  B\"uchi objectives correspond to
$\cParity{\{1,2\}}$, and co-B\"uchi objectives correspond to
$\cParity{\{0,1\}}$.
The objectives $\cParity{\{0,1,2\}}$ and $\cParity{\{1,2,3\}}$ are incomparable,
but they both subsume (modulo renaming of colors) B\"uchi and co-B\"uchi objectives.
Moreover, both $\cParity{\{0,1\}}$ and $\cParity{\{1,2\}}$ subsume the reachability objective $\reach{\reachset}$ (for MDPs with a sink~$\reachset$), by defining the color function so that $\coloring(s) = 1 \,\Leftrightarrow\, s \not\in \reachset$.
Similarly, both $\cParity{\{0,1\}}$ and $\cParity{\{1,2\}}$ subsume $\safety{\reachset}$, by defining $\coloring(s) = 1 \,\Leftrightarrow\, s \in \reachset$.

\medskip
\noindent {\bf Optimal and $\epsilon$-Optimal Strategies.}
Given an objective~$\formula$,
the value of state $\state$ in an MDP $\mdp$,
denoted by $\valueof{\mdp}{\state}$, is  
the supremum probability of achieving~$\formula$, i.e., 
$\valueof{\mdp}{\state}:=\sup_{\zstrat\in\zstratset}\probm_{\mdp,\state,\zstrat}(\formula)$. 
For $\epsilon \ge 0$ and $\state \in \states$, we say that
a strategy~$\zstrat$ is \emph{$\epsilon$-optimal}  iff
$\probm_{\mdp,\state,\zstrat}(\formula) \ge \valueof{\mdp}{\state}
-\epsilon$.
A $0$-optimal strategy is called \emph{optimal}. 
An optimal strategy is \emph{almost-surely winning} if $\valueof{\mdp}{\state}=1$.
Unlike in finite-state MDPs, optimal strategies need not exist
in countable MDPs, not even for reachability objectives in finitely branching MDPs.
However, by the definition of the value, for all 
$\epsilon>0$, an $\epsilon$-optimal strategy exists.

%A quantitative objective is a pair $\quantobj{\formula}{\constraint\const}$
%where $\formula$ is a reachability or parity objective and
%$\constraint\const$ (for $\const \in [0,1]$) is a probability constraint with~$\constraint\in\{\geq,>\}$.
For an objective~$\formula$ and $\constraint\in\{\geq,>\}$ and $\const \in [0,1]$,
we define $\zwinset{\formula}{\constraint\const}$ as the set of states~$\state$ for which there exists a strategy~$\zstrat$ with $\probm_{\mdp,\state,\zstrat}(\formula) \constraint \const$.
% we say that a strategy~$\zstrat$ is \emph{$\quantobj{\formula}{\constraint\const}$-winning}
%in $\state$ iff $\probm_{\mdp,\state,\zstrat}(\formula) \constraint \const$.
%The set of states with a $\quantobj{\formula}{\constraint\const}$-winning strategy is denoted by $\zwinset{\formula}{\constraint\const}$.
We call a state~$\state$ \emph{almost-surely winning} if $\state \in \zwinset{\formula}{\ge 1}$,
and we call $\state$ \emph{limit-surely winning} if $\state \in \zwinset{\formula}{\ge \const}$ for every constant $\const < 1$ (which is iff $\valueof{\mdp}{\state}=1$).
On infinite arenas, limit-surely winning states are not necessarily almost-surely winning.

\section{Objectives that require infinite memory}\label{sec-parity}

In this section we consider those objectives in the Mostowski hierarchy
where optimal (resp., $\epsilon$-optimal) strategies require infinite memory.
In each such case we construct an MDP that witnesses this requirement.
In these MDPs, all FR-strategies achieve the objective only with
probability~$0$, while some HD-strategy achieves the objective 
almost-surely (resp., with arbitrarily high probability).

 \begin{figure*}[t]
    \centering
    \begin{minipage}{.35\textwidth}
        \begin{center}				
        \centering
\begin{tikzpicture}[>=latex',shorten >=1pt,node distance=1.9cm,on grid,auto,
roundnode/.style={circle, draw,minimum size=1.5mm},
squarenode/.style={rectangle, draw,minimum size=2mm},
diamonddnode/.style={diamond, draw,minimum size=2mm}]

\node [squarenode,initial,initial text={}] (s0) at(0,0) [draw]{$s_0$};
\node [squarenode] (s1) [draw,right=1.7cm of s0]{$s_1$};
%\node [squarenode] (s2) [draw,right=1cm of s1]{$s_2$};
\node [squarenode] (s3) [draw=none,right=1.5cm of s1]{$\cdots$}; 
\node [squarenode] (s4) [draw,right=1.5cm of s3]{$s_i$}; 
\node [squarenode] (s5) [draw=none,right=1.5cm of s4]{$\cdots$}; 
 
\node[roundnode,double] (r0)  [below=1.4cm of s0] {$r_0$};
\node[roundnode,double] (r1)  [below=1.4cm of s1] {$r_1$};
%\node[roundnode,double] (r2)  [below=1.5cm of s2] {$r_2$};
\node[roundnode] (r3)  [draw=none,below=1.4 of s3] {$\cdots$};
\node[roundnode,double] (r4)  [below=1.4cm of s4] {$r_i$};
\node[roundnode,double] (r5)  [draw=none,below=1.4cm of s5] {$\cdots$};

\node [squarenode,double,inner sep = 4pt] (t)  [below=1.7cmof r1] {$ t $};
\node [rectangle, draw,inner sep = 2.4pt] (dum) [below=1.7cmof r1]{$ t $};

\path[->] (s0) edge (s1);
\path[->] (s1) edge (s3);
%\path[->] (s2) edge (s3);
\path[->] (s3) edge (s4);
\path[->] (s4) edge (s5);

\path[->] (s0) edge (r0);
\path[->] (s1) edge (r1);
%\path[->] (s2) edge (r2);
\path[->] (s4) edge (r4);

\path[->] (r0) edge node [near start,left=.05cm] {$1$} (t);
\path[->] (r1) edge node [near start,right=.05cm] {$\frac{1}{2}$} (t);
%\path[->] (r2) edge node [near start,right=0cm] {\scriptsize{$\frac{1}{4}$}} (t);
\path[->] (r4) edge node [near start,right=.25cm] {$\frac{1}{2^{i}}$} (t);

\path[->] (r1) edge [bend left=10] node[pos=0.2,below] {$\frac{1}{2}$} (s0);   
%\path[->] (r2) edge [bend left=5] node [pos=0.25,right=.25cm] {\scriptsize{$\frac{3}{4}$}} (s0);  
\path[->] (r4) edge [bend left=1] node [pos=0.1,above] {$1-\frac{1}{2^{i}}$} (s0);  

\draw [->] (t) -- ++(-2.2,0) --++(0,3.1)-- (s0);
\end{tikzpicture} 
				\end{center}
					\textbf{a)}~Almost-sure $\cParity{\{1,2,3\}}$
    \end{minipage}%
			\quad 
		\quad 
		\quad 
		\quad
		\quad 
		\quad
    \begin{minipage}{0.35\textwidth}
        \begin{center}
       \centering
\begin{tikzpicture}[>=latex',shorten >=1pt,node distance=1.9cm,on grid,auto,
roundnode/.style={circle, draw,minimum size=1.5mm},
squarenode/.style={rectangle, draw,minimum size=2mm},
diamonddnode/.style={diamond, draw,minimum size=2mm}]

\node [squarenode,double,initial,initial text={}] (s0) at(0,0) [draw]{$s_0$};
\node [squarenode] (s1) [draw,right=1.7cm of s0]{$s_1$};
%\node [squarenode] (s2) [draw,right=1cm of s1]{$s_2$};
\node [squarenode] (s3) [draw=none,right=1.5cm of s1]{$\cdots$}; 
\node [squarenode] (s4) [draw,right=1.5cm of s3]{$s_i$}; 
\node [squarenode] (s5) [draw=none,right=1.5cm of s4]{$\cdots$}; 
 
\node[roundnode] (r0)  [below=1.4cm of s0] {$r_0$};
\node[roundnode] (r1)  [below=1.4cm of s1] {$r_1$};
%\node[roundnode] (r2)  [below=1.5cm of s2] {$r_2$};
\node[roundnode] (r3)  [draw=none,below=1.4 of s3] {$\cdots$};
\node[roundnode] (r4)  [below=1.4cm of s4] {$r_i$};
\node[roundnode] (r5)  [draw=none,below=1.4cm of s5] {$\cdots$};

\node [squarenode] (t)  [below=1.7cmof r1] {$ b $};

\path[->] (s0) edge (s1);
\path[->] (s1) edge (s3);
%\path[->] (s2) edge (s3);
\path[->] (s3) edge (s4);
\path[->] (s4) edge (s5);

\path[->] (s0) edge (r0);
\path[->] (s1) edge (r1);
%\path[->] (s2) edge (r2);
\path[->] (s4) edge (r4);

\path[->] (r0) edge node [near start,left=.05cm] {$1$} (t);
\path[->] (r1) edge node [near start,right=.05cm] {$\frac{1}{2}$} (t);
%\path[->] (r2) edge node [near start,right=0cm] {\scriptsize{$\frac{1}{4}$}} (t);
\path[->] (r4) edge node [near start,right=.25cm] {$\frac{1}{2^{i}}$} (t);

\path[->] (r1) edge [bend left=10] node[pos=0.2,below] {$\frac{1}{2}$} (s0);   
%\path[->] (r2) edge [bend left=5] node [pos=0.25,right=.25cm] {\scriptsize{$\frac{3}{4}$}} (s0);  
\path[->] (r4) edge [bend left=1] node [pos=0.1,above] {$1-\frac{1}{2^{i}}$} (s0);

\path[->] (t) edge [loop right]   ();	
\end{tikzpicture}
				\end{center}
				\textbf{b)}~Limit-sure B\"uchi
    \end{minipage}
		\caption{Two finitely branching MDPs where the states~$s\in \zstates$  of the player 
				are drawn as squares and random states~$s\in \rstates$ as circles.
				The color~$\colorof\state$ of~$\state$ is indicated with the number of boundaries; 
         for example, a double boundary for color~$2$.
				State~$s_0$ in the MDP on the left is almost-surely winning for $\cParity{\{1,2,3\}}$, but 
				all almost-surely winning strategies require infinite memory.
				The MDP on the right is such that, for all $c>0$,  strategies that achieve B\"uchi with probability at least~$c$  require infinite memory.}
\label{fig:123MotwRequiresInfMemBuchiInfMem}
\end{figure*}
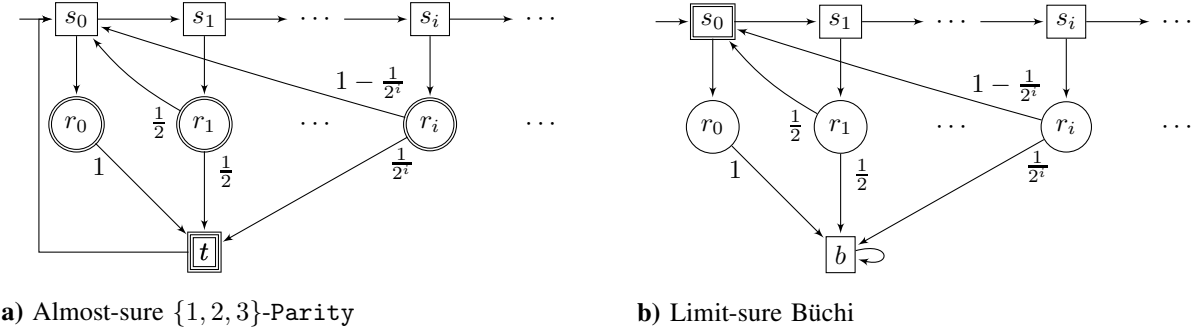

% \medskip

% For $\cParity{\{1,2,3\}}$, we show:

\newcommand{\thmparityonetwothree}{
Let $\formula=\cParity{\{1,2,3\}}$. There exists a finitely branching MDP~$\mdp$ with initial state $\state_0$ such that 
\begin{itemize}
	\item for all FR-strategies~$\zstrat$, we have $\probm_{\mdp,\state_0,\zstrat}(\formula)=0$,
	\item  there exists an HD-strategy $\zstrat$ such that  $\probm_{\mdp,\state_0,\zstrat}(\formula)=1$. 
\end{itemize}
Hence, optimal  (and even almost-surely winning)  and $\epsilon$-optimal strategies   require infinite memory for~$\cParity{\{1,2,3\}}$, 
even in finitely branching MDPs.
}

\begin{theorem}\label{thm:parity123}
\thmparityonetwothree
\end{theorem}

The MDP in Theorem~\ref{thm:parity123} is depicted 
in Figure~\ref{fig:123MotwRequiresInfMemBuchiInfMem} (left),
% The initial state is~$s_0$, the color function is 
where $\coloring(s_i)=1$ and $\coloring(r_i)=2$ for all $i\in \nat$, and $\coloring(t)=3$.
For every FR-strategy there is a uniform lower bound on the probability of 
visiting~$t$ between consecutive visits to~$s_0$.
Hence, unless the strategy with positive probability eventually
always stays in states $s_i$ (and
thus also loses the almost-sure parity objective),
in the long-run, the probability of visiting $t$ (with color three) tends to~$1$, 
and the parity condition is satisfied with probability~$0$.
Although the player cannot win by  any FR-strategy,  
we construct an HD-strategy $\zstrat$ such that $\probm_{\mdp,\state_0,\zstrat}(\formula)=1$.
This strategy is such that 
upon the $i^{\text{th}}$ visit to~$\state_0$,  
the ladder $\state_0 \state_1 \cdots \state_{i}$ is traversed and  the 
transition~$\state_{i} \transition r_{i}$ is chosen.
Moving further along the ladder  $\state_0 \state_1 \state_2 \cdots$ 
decreases the probability of visiting~$t$ between the previous and successive visits to~$s_0$.
Hence, the probability of visiting 
color three infinitely often is~$0$.

\newcommand{\remarkparityonecounter}{
A strict subclass of finitely branching MDPs are $1$-counter MDPs,
where a finite-state MDP is augmented with an integer counter
\cite{BBEKW10}.
The MDP in Theorem~\ref{thm:parity123} (plus some auxiliary states)
is implementable by a $1$-counter MDP.
}
\begin{remark}\label{rem:parityonecounter}
\remarkparityonecounter
\end{remark}

\newcommand{\remarkparityRabinStreet}{
The classical Rabin and Streett conditions can encode $\cParity{\{1,2,3\}}$.
Thus, optimal and $\epsilon$-optimal strategies for Rabin/Streett  require infinite memory, 
even in finitely branching countable MDPs.

On finite MDPs, optimal strategies can be chosen MD for parity and Rabin
objectives, but not for Streett objectives.
Optimal strategies for Streett objectives can be chosen MR or FD \cite{CAH:QEST2004}. 
}
\begin{remark}\label{rem:parityRabinStreet}
\remarkparityRabinStreet
\end{remark}
\begin{proof}

For an infinite play~$\pi^{\infty}$, let ${\sf Inf}(\pi^{\infty})$ be the set of states that~$\pi^{\infty}$ visits infinitely often.
Let us recall the Rabin and Streett conditions.

Given a Rabin condition~$\{(E_1,F_1), (E_2,F_2),\cdots, (E_n,F_n)\}$ with~$n$ pairs (or $n$ disjunctions),
an infinite play~$\pi^{\infty}$ satisfies the Rabin condition if there exists a pair~$(E_i,F_i)$ such that 
${\sf Inf}(\pi^{\infty})\cap E_i= \emptyset$ and ${\sf Inf}(\pi^{\infty})\cap F_i\neq \emptyset$.
The Rabin condition 
$$\{(\colorset{\states}{}{=3},\colorset{\states}{}{=2})\}$$
 encodes~$\cParity{\{1,2,3\}}$, since 
all satisfying runs must visit states with color~$2$ infinitely often
and states with color~$3$ only finitely often.
Note that~$\cParity{\{1,2,3\}}$ is encoded in a Rabin condition with only one disjunction.

Given a Streett condition~$\{(E_1,F_1), (E_2,F_2),\cdots, (E_n,F_n)\}$ with~$n$ pairs (or $n$ conjunctions),
an infinite play~$\pi^{\infty}$ satisfies the Streett condition if 
${\sf Inf}(\pi^{\infty})\cap E_i=\emptyset$ implies ${\sf Inf}(\pi^{\infty})\cap F_i= \emptyset$ for all pairs~$(E_i,F_i)$.
The Streett condition 
$$\{(\colorset{\states}{}{=2},\states),(\emptyset,\colorset{\states}{}{=3})\}$$
encodes $\cParity{\{1,2,3\}}$, since
all satisfying runs must visit states with color~$2$ infinitely often
and states with color~$3$ only finitely often.

Note that a conjunction of \emph{two} Streett pairs are needed to encode
$\cParity{\{1,2,3\}}$. A single Streett pair $\{(X,Y)\}$ means
``infinitely often $X$ or only finitely often $Y$'', which can be encoded
as a $\cParity{\{0,1,2\}}$ condition by assigning color $2$ to $X$ and color
$1$ to $Y$. Unlike for $\cParity{\{1,2,3\}}$, optimal strategies for
$\cParity{\{0,1,2\}}$ (and thus also for a single Streett pair)
can be chosen MD in finitely branching MDPs 
(Theorem~\ref{thm:012quant}).
\end{proof}

\smallskip
\noindent
It was known that quantitative B\"uchi objectives require infinite memory \cite{Krcal:Thesis:2009,BBS:ACM2007}. 
For the sake of completeness, we present an example MDP for Proposition~\ref{thm:wpp-Buchi}
in Figure~\ref{fig:123MotwRequiresInfMemBuchiInfMem} (right).

\newcommand{\thmwppBuchi}{
Let $\formula=\cParity{\{1,2\}}$ be the B\"uchi  objective. 
There exists a finitely branching MDP~$\mdp$ with initial state $\state_0$ such that 
\begin{itemize}
	\item for all FR-strategies~$\zstrat$, we have $\probm_{\mdp,\state_0,\zstrat}(\formula)=0$,
	\item   for every~$c\in [0,1)$, there exists an HD-strategy $\zstrat$ such that  $\probm_{\mdp,\state_0,\zstrat}(\formula)\geq c$. 
\end{itemize}
Hence, $\epsilon$-optimal strategies for B\"uchi objectives require infinite memory.  
% In particular,  all families of strategies, that are winning for the limit-sure B\"uchi,
% are families of infinite-memory strategies.
}

\begin{proposition}[\cite{Krcal:Thesis:2009}]\label{thm:wpp-Buchi}
\thmwppBuchi
\end{proposition}

% \bigskip
% For safety, we show:

\newcommand{\thminfBranchSafety}{
Let $\formula=\safety{\reachset}$. 
There exists an infinitely branching MDP~$\mdp$ with initial state $\state$ such that
\begin{itemize}
	\item for all FR-strategies~$\zstrat$, we have $\probm_{\mdp,\state,\zstrat}(\formula)=0$,
	\item  for every~$c\in [0,1)$, there exists an HD-strategy $\zstrat$ such that  $\probm_{\mdp,\state,\zstrat}(\formula)\geq c$. 
\end{itemize}
Hence, $\epsilon$-optimal  strategies for safety require infinite memory. 
}

\begin{theorem}\label{thm:infBranchSafety}
\thminfBranchSafety
\end{theorem}
The MDP in Theorem~\ref{thm:infBranchSafety}, depicted in
Figure~\ref{fig:infBranchCoBuchiMaxSafety} (left), 
was first introduced in~\cite{Kucherabook11}.
Since our notion of finite-memory strategies allows for randomized memory updates (in contrast to~\cite{Kucherabook11}),
 our proof is somewhat more general.
The target is $\reachset=\{t\}$.
For every FR-strategy there is a uniform lower bound on the probability of
reaching~$t$ between consecutive visits to~$s_0$. Since $t$ is absorbing, it
will be reached with probability~$1$. Thus every FR-strategy satisfies the
safety objective with probability $0$.  
However, for all~$n\in \nat$,
we construct 
an HD-strategy $\zstrat_n$ such 
that $\probm_{\mdp,\state,\zstrat_n}(\safety{\{t\}})\geq 1-\frac{1}{2^n}$.
This strategy is such that  upon the $i^{\text{th}}$ visit to~$s$,  the 
transition~$s \transition r_{i+n}$ is chosen. 
Hence, the probability of visiting~$t$   between two successive visits to~$s$ decreases.
A more detailed analysis shows that the probability of ever visiting~$t$ 
is bounded by~$\frac{1}{2^n}$. 
%

% \bigskip
% For co-B\"uchi, we show:

\newcommand{\thminfBranchCoBuchi}{
Let $\formula=\cParity{\{0,1\}}$ be the co-B\"uchi objective.
There exists an infinitely branching MDP~$\mdp$ with initial state $\state$ such that
\begin{itemize}
	\item for all FR-strategies~$\zstrat$, we have $\probm_{\mdp,\state,\zstrat}(\formula)=0$,
	\item  there exists an HD-strategy $\zstrat$ such that  $\probm_{\mdp,\state,\zstrat}(\formula)=1$. 
\end{itemize}
Hence, optimal (and even almost-surely winning) strategies and $\epsilon$-optimal strategies for co-B\"uchi require infinite memory. 
}

\begin{theorem}\label{thm:infBranchCoBuchi}
\thminfBranchCoBuchi
\end{theorem}

The MDP in Theorem~\ref{thm:infBranchCoBuchi} is depicted in
Figure~\ref{fig:infBranchCoBuchiMaxSafety} (right).
By a similar  argument as in Theorem~\ref{thm:infBranchSafety}, 
every FR-strategy achieves co-B\"uchi with probability $0$.  
However, the HD-strategy $\zstrat$ that chooses the 
transition~$s \transition r_{i}$ upon the $i^{\text{th}}$ visit to~$s$
is  such that $\probm_{\mdp,\state,\zstrat}(\formula)=1$.

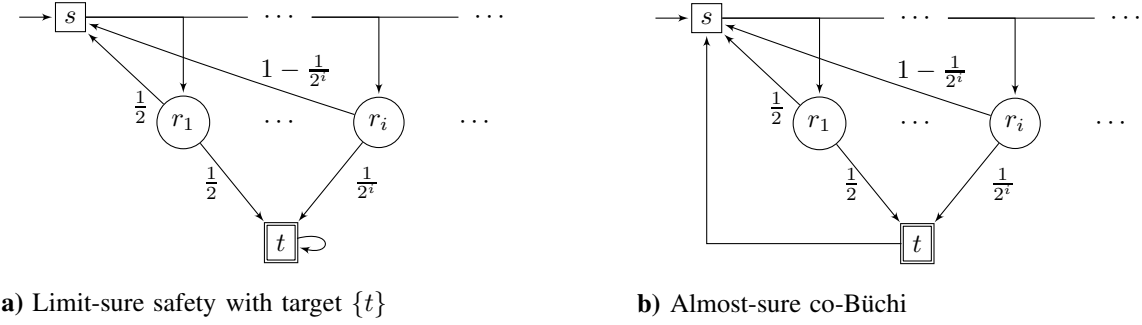
\begin{figure*}[t]
    \centering
    \begin{minipage}{.35\textwidth}
        \begin{center}
								\centering
\begin{tikzpicture}[>=latex',shorten >=1pt,node distance=1.9cm,on grid,auto,
roundnode/.style={circle, draw,minimum size=1.5mm},
squarenode/.style={rectangle, draw,minimum size=2mm},
diamonddnode/.style={diamond, draw,minimum size=2mm}]

\node [squarenode,initial,initial text={}] (s) at(0,0) [draw]{$s$};

\node[roundnode] (r1)  [below right=1.4cm and 1.5cmof s0] {$r_1$};
%\node[roundnode] (r2)  [right=1cm of r1] {$r_2$};
\node[roundnode] (r3)  [draw=none,right=1.3 of r1] {$\cdots$};
\node[roundnode] (r4)  [right=1.3cm of r3] {$r_i$};
\node[roundnode] (r5)  [draw=none,right=1.3cm of r4] {$\cdots$};

\node [squarenode,double,inner sep = 4pt] (t)  [below=1.6cmof r3] {$ t $};

\draw [->] (s) -- ++(1.5,0) -- (r1);
%\draw [->] (s) -- ++(2,0) -- (r2);
\node[roundnode] (d)  [draw=none,right=2.8 of s] {$\cdots$};
\draw[-] (s) -- (d);
\draw [->] (d) -- ++(1.3,0) -- (r4);
\node[roundnode] (dd)  [draw=none,right=2.8 of d] {$\cdots$};
\draw[-] (d) -- (dd);

\path[->] (r1) edge node [midway,left=.05cm] {$\frac{1}{2}$} (t);
%\path[->] (r2) edge node [near start,right=0cm] {\scriptsize{$\frac{1}{4}$}} (t);
\path[->] (r4) edge node [midway,right=.2cm] {$\frac{1}{2^{i}}$} (t);
\path[->] (r1) edge  node[pos=0.3,below] {$\frac{1}{2}$} (s);   
%\path[->] (r2) edge [bend left=5] node [pos=0.04,above] {\scriptsize{$\frac{3}{4}$}} (s);  
\path[->] (r4) edge [bend left=1] node [pos=0.2,above] {$1-\frac{1}{2^{i}}$} (s);  

\path[->] (t) edge [loop right]  (t);
\end{tikzpicture}
				\end{center}
				\textbf{a)}~Limit-sure safety with target~$\{t\}$
        
    \end{minipage}%
		\quad 
		\quad 
		\quad 
		\quad
		\quad 
		\quad
    \begin{minipage}{0.35\textwidth}
        \begin{center}
				\centering
\begin{tikzpicture}[>=latex',shorten >=1pt,node distance=1.9cm,on grid,auto,
roundnode/.style={circle, draw,minimum size=1.5mm},
squarenode/.style={rectangle, draw,minimum size=2mm},
diamonddnode/.style={diamond, draw,minimum size=2mm}]

\node [squarenode,initial,initial text={}] (s) at(0,0) [draw]{$s$};

\node[roundnode] (r1)  [below right=1.4cm and 1.5cmof s0] {$r_1$};
%\node[roundnode] (r2)  [right=1cm of r1] {$r_2$};
\node[roundnode] (r3)  [draw=none,right=1.3 of r1] {$\cdots$};
\node[roundnode] (r4)  [right=1.3cm of r3] {$r_i$};
\node[roundnode] (r5)  [draw=none,right=1.3cm of r4] {$\cdots$};

\node [squarenode,double,inner sep = 4pt] (t)  [below=1.6cmof r3] {$ t $};

\draw [->] (s) -- ++(1.5,0) -- (r1);
%\draw [->] (s) -- ++(2,0) -- (r2);
\node[roundnode] (d)  [draw=none,right=2.8 of s] {$\cdots$};
\draw[-] (s) -- (d);
\draw [->] (d) -- ++(1.3,0) -- (r4);
\node[roundnode] (dd)  [draw=none,right=2.8 of d] {$\cdots$};
\draw[-] (d) -- (dd);

\path[->] (r1) edge node [midway,left] {$\frac{1}{2}$} (t);
%\path[->] (r2) edge node [near start,right=0cm] {\scriptsize{$\frac{1}{4}$}} (t);
\path[->] (r4) edge node [midway,right=.2cm] {$\frac{1}{2^{i}}$} (t);
\path[->] (r1) edge  node[pos=0.3,below] {$\frac{1}{2}$} (s);   
%\path[->] (r2) edge [bend left=5] node [pos=0.04,above] {\scriptsize{$\frac{3}{4}$}} (s);  
\path[->] (r4) edge [bend left=1] node [pos=0.2,above] {$1-\frac{1}{2^{i}}$} (s);  

\draw [->] (t) -- ++(-2.8,0) -- (s);
\end{tikzpicture}	
				\end{center}
					\textbf{b)}~Almost-sure co-B\"uchi
    \end{minipage}
		\caption{In the infinitely branching MDP on the left, all $\epsilon$-optimal strategies for Safety  require infinite memory.
        In the infinitely branching MDP on the right, all optimal (and thus almost-surely winning)
				strategies for  co-B\"uchi   require infinite memory.}
\label{fig:infBranchCoBuchiMaxSafety}
\end{figure*}

\section{From almost-sure winning to optimal strategies}
\label{sec-almost-to-optimal}
In this section we prove Theorem~\ref{thm:reduction-to-as}.
It says that, for certain objectives, if 
almost-surely winning strategies (where they exist) can be chosen MD,
then optimal strategies (where they exist) can also be chosen MD.

We call a class~$\mdpclass$ of MDPs \emph{downward-closed} if every MDP whose transition relation is a subset of the transition relation of some MDP in~$\mdpclass$ is also in~$\mdpclass$.
The class of finitely branching MDPs is downward-closed, and so is the class of MDPs with a fixed sink~$\reachset$.

We call an objective~$\formula$ \emph{prefix-independent in $\mdpclass$} (where $\mdpclass$ is a class of MDPs) if for all $w_1, w_2 \in \states^*$ and all $w \in \states^\omega$ such that $w_1 w$ and $w_2 w$ are infinite plays in an MDP in~$\mdpclass$ we have $w_1 w \in \denotationof{\formula}{} \Longleftrightarrow w_2 w \in \denotationof{\formula}{}$.
Parity objectives are prefix-independent in the class of all MDPs.
Both objectives $\reach{\reachset}$ and $\safety{\reachset}$ are prefix-independent in the class of MDPs with sink~$\reachset$.

The following theorem provides, under certain conditions, an optimal MD-strategy for all states that have an optimal strategy.
In fact, a single MD-strategy is optimal for all states that have an optimal strategy:
\newcommand{\stmtthmreductiontoas}{
Let $\formula$ be an objective that is prefix-independent in a downward-closed class~$\mdpclass$ of MDPs.
Suppose that for any $\mdp=\mdptuple \in \mdpclass$ and any $s \in \states$ and any strategy~$\zstrat$ with $\probm_{\mdp,\state,\zstrat}(\formula) = 1$ there exists an MD-strategy~$\zstrat'$ with $\probm_{\mdp,\state,\zstrat'}(\formula) = 1$.

%Under this condition, for $\mdp \in \mdpclass$, if there exists an optimal strategy for $\formula$,
%then there also exists an optimal MD-strategy for $\formula$. Formally,
%if $\probm_{\mdp,\state,\zstrat}(\formula) = \valueof{\mdp}{\state}$ then
%there exists an MD-strategy~$\zstrat'$ with $\probm_{\mdp,\state,\zstrat'}(\formula) = \valueof{\mdp}{\state}$.
Under this condition, for each $\mdp \in \mdpclass$
there is an MD-strategy~$\zstrat'$ such that for all $\state \in \states$:
\[
\begin{aligned}
\big(
\exists \zstrat \in \zstratset.\,
\probm_{\mdp,\state,\zstrat}(\formula) = \valueof{\mdp}{s}
\big) & \Longrightarrow \\
& \ \probm_{\mdp,\state,\zstrat'}(\formula) = \valueof{\mdp}{s}
\end{aligned}
\]
}
\begin{theorem} \label{thm:reduction-to-as}
\stmtthmreductiontoas
\end{theorem}

%
%By Proposition~\ref{prop:as012}, the winning condition $\formula=\cParity{\{0,1,2\}}$ has this property.
%Hence it follows:
%
%\begin{theorem} \label{thm:012quant}
%Let $\mdp=\mdptuple$ be an MDP, and $\state_0 \in \states$, and $\zstrat$ a strategy, and $\coloring:\states\to \{0,1,2\}$, and $\formula=\cParity{\{0,1,2\}}$.
%Suppose $\probm_{\mdp,\state_0,\zstrat}(\formula) = \valueof{\mdp}{\state_0}$.
%Then there is an MD-strategy~$\zstrat'$ with $\probm_{\mdp,\state_0,\zstrat'}(\formula) = \valueof{\mdp}{\state_0}$.
%In other words, for $\cParity{\{0,1,2\}}$ objectives, MD-strategies suffice to play optimally.
%\end{theorem}
%
The remainder of the section is devoted to the proof of Theorem~\ref{thm:reduction-to-as}. 

%The following lemma states---for prefix-independent winning conditions---that whenever an optimal strategy visits some state, it achieves this state's value:

\newcommand{\lemprefixindoptimality}{
Let $\formula$ be an objective that is prefix-independent in a class~$\mdpclass$ of MDPs.
Let $\mdp=\mdptuple \in \mdpclass$, and $s_0 \in \states$, and $\zstrat$ be a strategy with $\probm_{\mdp,s_0,\zstrat}(\formula) = \valueof{\mdp}{s_0}$.
Suppose that $s_0 s_1 \cdots s_n$ for some $n \ge 0$ is a partial play starting in $s_0$ and induced by~$\zstrat$.
Then:
\begin{enumerate}
\item $\valueof{\mdp}{s_n} = \probm_{\mdp,s_0,\zstrat}(\denotationof{\formula}{s_0} \mid s_0 s_1 \cdots s_n \states^\omega)$.
\item If $s_n \in \rstates$ then $\valueof{\mdp}{s_n} = \sum_{s_{n+1} \in \states} \probp(s_n)(s_{n+1}) \cdot \valueof{\mdp}{s_{n+1}}$.
\item If $s_n \in \zstates$ then $\valueof{\mdp}{s_n} = \valueof{\mdp}{s_{n+1}}$ for all $s_{n+1} \in \supp(\sigma(s_0 s_1 \cdots s_n))$.
\end{enumerate}
}
%\begin{lemma} \label{lem:prefix-ind-optimality}
%\lemprefixindoptimality
%\end{lemma}

\newcommand{\pmdp}{\mdp_{*}}
\newcommand{\pstates}{\states_{*}}
\newcommand{\pzstates}{\states_{*\zsymbol}}
\newcommand{\prstates}{\states_{*\rsymbol}}
\newcommand{\ptransition}{\transition_{*}}
\newcommand{\pprobp}{\probp_{*}}

%We use Lemma~\ref{lem:prefix-ind-optimality} to show that the MDP constructed in the following lemma is well-defined.
%This MDP, $\pmdp$, will be crucial for the proof of Theorem~\ref{thm:reduction-to-as}.
%Loosely speaking, $\pmdp$ is the MDP~$\mdp$ conditioned under~$\formula$.

\medskip

For prefix-independent winning conditions, whenever an optimal strategy visits some state, it achieves the value of this state; see Lemma~\ref{lem:prefix-ind-optimality} in the appendix.
We use this to show that the MDP constructed in the following lemma is well-defined.
This MDP, $\pmdp$, will be crucial for the proof of Theorem~\ref{thm:reduction-to-as}.
Loosely speaking, $\pmdp$ is the MDP~$\mdp$ conditioned under~$\formula$.

\newcommand{\lemconditionedconstruction}{
Let $\formula$ be an objective that is prefix-independent in a class~$\mdpclass$ of MDPs.
Let $\mdp=\mdptuple \in \mdpclass$. %, and $s_0 \in \states$, and $\zstrat$ be a strategy with $\probm_{\mdp,s_0,\zstrat}(\formula) = \valueof{\mdp}{s_0} > 0$.
Construct an MDP $\pmdp = \tuple{\pstates,\pzstates,\prstates,\ptransition,\pprobp}$ by setting
\[
\begin{aligned}
\pstates = \{s \in \states \mid \exists\,\zstrat.\; \probm_{\mdp,s,\zstrat}(\formula) = \valueof{\mdp}{s} > 0\}
\end{aligned}
\]
and $\pzstates = \pstates \cap \zstates$ and $\prstates = \pstates \cap \rstates$ and
\[
\begin{aligned}
\mathord{\ptransition} = 
\{(s,t) \in \pstates \times \pstates \mid {} & s \transition t \text{ and if $s \in \pzstates$}\\
&\text{ then $\valueof{\mdp}{s} = \valueof{\mdp}{t}$}\}
\end{aligned}
\]
and $\pprobp : \prstates \to \dist(\pstates)$ so that
\[
\begin{aligned}
\pprobp(s)(t) = \probp(s)(t) \cdot \frac{\valueof{\mdp}{t}}{\valueof{\mdp}{s}}
\end{aligned}
\]
for all $s \in \prstates$ and $t \in \pstates$ with $s \, \ptransition \, t$. Then:
\begin{enumerate}
\item
For all $\zstrat \in \zstratset_{\pmdp}$ and all $n \ge 0$ and all $s_0, \ldots, s_n \in \pstates$ with $s_0 \ptransition\, s_1 \ptransition\, \cdots \ptransition\, s_n$:
\[
\begin{aligned}
\probm_{\pmdp,s_0,\zstrat}&(s_0 s_1 \cdots s_n \states^\omega) = \\
&\probm_{\mdp,s_0,\zstrat}(s_0 s_1 \cdots s_n \states^\omega) \cdot \frac{\valueof{\mdp}{s_n}}{\valueof{\mdp}{s_0}}
\end{aligned}
\]
\item 
For all $s_0 \in \pstates$ and all $\zstrat \in \zstratset_{\mdp}$ with $\probm_{\mdp,s_0,\zstrat}(\formula) = \valueof{\mdp}{s_0} > 0$ and all measurable $\playset \subseteq s_0 \states^\omega$ we have
$
\probm_{\pmdp,s_0,\zstrat}(\playset) = \probm_{\mdp,s_0,\zstrat}(\playset \mid \denotationof{\formula}{s_0})
$.
\end{enumerate}
}

\begin{lemma}\label{lem:conditioned-construction}
\lemconditionedconstruction
\end{lemma}

The following lemma provides, under certain conditions, a \emph{uniform} almost-surely winning MD-strategy, i.e., one that works for all initial states at the same time:

\newcommand{\lemMDPasuniform}{
Let $\mdp=\mdptuple$ be an MDP.
Let $\formula$ be an objective that is prefix-independent in~$\{\mdp\}$.
Suppose that for any $s \in \states$ and any strategy~$\zstrat$ with $\probm_{\mdp,\state,\zstrat}(\formula) = 1$ there exists an MD-strategy~$\zstrat'$ with $\probm_{\mdp,\state,\zstrat'}(\formula) = 1$.
Then there is an MD-strategy~$\zstrat'$ such that for all $\state \in \states$:
\[
\big(
\exists \zstrat \in \zstratset.\,
\probm_{\mdp,\state,\zstrat}(\formula) = 1
\big)
\quad\Longrightarrow\quad
\probm_{\mdp,\state,\zstrat'}(\formula) = 1
\]
}

\begin{lemma}\label{lem:MDP-as-uniform}
\lemMDPasuniform
\end{lemma}
\begin{proof}
We can assume that all states are almost-surely winning, since in order to achieve an almost-sure winning objective, the player must forever remain in almost-surely winning states.
So we need to define an MD-strategy~$\zstrat'$ so that for all $s \in \states$ we have $\probm_{\mdp,\state,\zstrat'}(\formula) = 1$.
%In this proof, we sometimes speak about partially defined MD-strategies, i.e., functions $\sigma : V \cap \zstates \to \dist(\states)$ for subsets $V \subseteq \states$.
%Then, for $\playset \subseteq V^\omega$, we may write $\probm_{\mdp,\state,\zstrat}(\playset)$ to mean $\probm_{\mdp,\state,\zstrat'}(\playset)$ for a strategy~$\zstrat'$ that extends~$\zstrat$ in an arbitrary way.

Fix an arbitrary state $s_1 \in \states$.
By assumption there is an MD-strategy~$\zstrat_1$ with $\probm_{\mdp,s_1,\zstrat_1}(\formula) = 1$.
Let $U_1 \subseteq \states$ be the set of states that occur in plays that both start from~$s_1$ and are induced by~$\zstrat_1$.
We have $\probm_{\mdp,s_1,\zstrat_1}(\denotationof{\formula}{s_1} \cap U_1^\omega) = 1$.
In fact, for any $s \in U_1$ and any strategy~$\zstrat$ that agrees with~$\zstrat_1$ on~$U_1$ we have $\probm_{\mdp,s,\zstrat}(\denotationof{\formula}{s} \cap U_1^\omega) = 1$.

If $U_1=\states$ we are done.
Otherwise, consider the MDP~$\mdp_1$ obtained from~$\mdp$ by fixing~$\zstrat_1$ on~$U_1$ (i.e., in~$\mdp_1$ we can view the states in~$U_1$ as random states).
We argue that, in~$\mdp_1$, for any state~$s$ there is an MD-strategy~$\zstrat_1'$ with $\probm_{\mdp_1,\state,\zstrat_1'}(\formula) = 1$.
Indeed, let $s \in \states$ be any state.
Recall that there is an MD-strategy~$\zstrat$ with $\probm_{\mdp,\state,\zstrat}(\formula) = 1$.
Let $\zstrat_1'$ be the MD-strategy obtained by restricting~$\zstrat$ to the non-$U_1$ states (recall that the $U_1$ states are random states in~$\mdp_1$).
This strategy~$\zstrat_1'$ almost surely generates a play that \emph{either} satisfies~$\formula$ without ever entering~$U_1$ \emph{or} at some point enters~$U_1$.
In the latter case, $\formula$ is satisfied almost surely: this follows from prefix-independence and the fact that $\zstrat_1'$ agrees with~$\zstrat_1$ on~$U_1$.
We conclude that $\probm_{\mdp_1,\state,\zstrat_1'}(\formula) = 1$.

Let $s_2 \in \states \setminus U_1$.
We repeat the argument from above, with $s_2$ instead of~$s_1$, and with $\mdp_1$ instead of~$\mdp$.
This yields an MD-strategy~$\zstrat_2$ and a set $U_2 \ni s_2$ with $\probm_{\mdp_1,s_2,\zstrat_2}(\denotationof{\formula}{s_2} \cap U_2^\omega) = 1$.
In fact, for any $s \in U_2$ and any strategy~$\zstrat$ that agrees with~$\zstrat_2$ on~$U_2$ and with~$\zstrat_1$ on~$U_1$ we have $\probm_{\mdp,s,\zstrat}(\denotationof{\formula}{s} \cap U_2^\omega) = 1$.

If $U_1 \cup U_2=\states$ we are done.
Otherwise we continue in the same manner, and so forth.
Since $\states$ is countable, we can pick $s_1, s_2, \ldots$ to have $\bigcup_{i \ge 1} U_i= \states$.
Define an MD-strategy~$\zstrat'$ such that for any $s \in \zstates$ we have $\zstrat'(s) = \zstrat_i(s)$ for the smallest~$i$ with $s \in U_i$.
Thus, if $s \in U_i$, we have $\probm_{\mdp,s,\zstrat'}(\formula) \ge \probm_{\mdp,s,\zstrat'}(\denotationof{\formula}{s} \cap U_i^\omega) = 1$.
\end{proof}

\newcommand{\cyl}{\mathfrak C}
\newcommand{\classcyl}{\mathcal{C}}
\newcommand{\classmon}{\mathcal{Q}}

The following measure-theoretic lemma will be used to connect probability measures induced by the MDPs $\mdp$ and $\pmdp$ from Lemma~\ref{lem:conditioned-construction}.

\newcommand{\lemmeasuretheory}{
Let $\states$ be countable and $s \in \states$.
Call a set of the form $s w \states^\omega$ for $w \in \states^*$ a \emph{cylinder}.
Let $\probm, \probm'$ be probability measures on $s S^\omega$ defined in the standard way, i.e., first on cylinders and then extended to all measurable sets $\playset \subseteq s \states^\omega$.
Suppose there is $x \ge 0$ such that $x \cdot \probm(\cyl) \le \probm'(\cyl)$ for all cylinders~$\cyl$.
Then $x \cdot \probm(\playset) \le \probm'(\playset)$ holds for all measurable $\playset \subseteq s S^\omega$.
}

\begin{lemma} \label{lem:measure-theory}
\lemmeasuretheory
\end{lemma}

We are ready to prove Theorem~\ref{thm:reduction-to-as}.

\begin{proof}[Proof of Theorem~\ref{thm:reduction-to-as}]
As in the statement of the theorem, suppose that $\formula$ is an objective that is prefix-independent in a downward-closed class~$\mdpclass$ of MDPs so that for any $\mdp=\mdptuple \in \mdpclass$ and any $s \in \states$ and any strategy~$\zstrat$ with $\probm_{\mdp,\state,\zstrat}(\formula) = 1$ there exists an MD-strategy~$\zstrat'$ with $\probm_{\mdp,\state,\zstrat'}(\formula) = 1$.
Let $\mdp=\mdptuple \in \mdpclass$. %, and $s_0 \in \states$, and $\zstrat$ be a strategy with $\probm_{\mdp,s_0,\zstrat}(\formula) = \valueof{\mdp}{s_0}$.
%We can assume $\valueof{\mdp}{s_0} > 0$.
Let $\pmdp = \tuple{\pstates,\pzstates,\prstates,\ptransition,\pprobp}$ be the MDP defined in Lemma~\ref{lem:conditioned-construction}.
Since $\mdpclass$ is downward-closed, we have $\pmdp \in \mdpclass$.
In particular, $\formula$ is prefix-independent in~$\{\pmdp\}$.

First we show that for any $s \in \pstates$ there exists an MD-strategy~$\zstrat'$ with $\probm_{\pmdp,s,\zstrat'}(\formula) = 1$.
Indeed, let $s \in \pstates$.
By the definition of~$\pstates$, there is a strategy~$\zstrat$ with $\probm_{\mdp,s,\zstrat}(\formula) = \valueof{\mdp}{s} > 0$.
By Lemma~\ref{lem:conditioned-construction}.2, we have $\probm_{\pmdp,s,\zstrat}(\formula) = 1$.
By our assumption on~$\mdpclass$ there exists an MD-strategy~$\zstrat'$ with $\probm_{\pmdp,s,\zstrat'}(\formula) = 1$.

By Lemma~\ref{lem:MDP-as-uniform}, it follows that there is an MD-strategy~$\zstrat'$ with $\probm_{\pmdp,s,\zstrat'}(\formula) = 1$ for all $s \in \pstates$.
We show that this strategy~$\zstrat'$ satisfies the property claimed in the statement of the theorem.

To this end, let $n \ge 0$ and $s_0, s_1, \ldots, s_n \in \states$.
If $s_0 s_1 \cdots s_n$ is a partial play in~$\pmdp$ then, by Lemma~\ref{lem:conditioned-construction}.1, 
\begin{align*}
&\ \probm_{\pmdp,s_0,\zstrat'}(s_0 s_1 \cdots s_n \states^\omega) \\ 
= &\ \probm_{\mdp,s_0,\zstrat'}(s_0 s_1 \cdots s_n \states^\omega) \cdot \frac{\valueof{\mdp}{s_n}}{\valueof{\mdp}{s_0}} \,,
\intertext{
and thus, as $\valueof{\mdp}{s_n} \le 1$,
}
&\ \valueof{\mdp}{s_0} \cdot \probm_{\pmdp,s_0,\zstrat'}(s_0 s_1 \cdots s_n \states^\omega) \\
\le &\ \probm_{\mdp,s_0,\zstrat'}(s_0 s_1 \cdots s_n \states^\omega)\,.
\end{align*}
If $s_0 s_1 \cdots s_n$ is not a partial play in~$\pmdp$ then $\probm_{\pmdp,s_0,\zstrat'}(s_0 s_1 \cdots s_n \states^\omega) = 0$ and the previous inequality holds as well.
Therefore, by Lemma~\ref{lem:measure-theory}, we get for all measurable sets $\playset \subseteq s_0 \states^\omega$:
\[
\valueof{\mdp}{s_0} \cdot \probm_{\pmdp,s_0,\zstrat'}(\playset) \le \probm_{\mdp,s_0,\zstrat'}(\playset)
\]
In particular, since $\probm_{\pmdp,s_0,\zstrat'}(\formula) = 1$, we obtain $\valueof{\mdp}{s_0} \le \probm_{\mdp,s_0,\zstrat'}(\formula)$.
The converse inequality $\probm_{\mdp,s_0,\zstrat'}(\formula) \le \valueof{\mdp}{s_0}$ holds by the definition of~$\valueof{\mdp}{s_0}$, hence we conclude $\probm_{\mdp,s_0,\zstrat'}(\formula) = \valueof{\mdp}{s_0}$.
\end{proof}

\section{When MD-strategies suffice in general countable MDPs}\label{sec-inf-MDPs}

Ornstein~\cite{Ornstein:AMS1969} shows that $\epsilon$-optimal and optimal strategies for reachability can be chosen MD:

\begin{theorem}[from Theorem B in \cite{Ornstein:AMS1969}]\label{thm:reach-eps}
For every countable MDP $\mdp$ there exist uniform $\epsilon$-optimal MD-strategies for
reachability objectives $\formula = \reach{\reachset}$, i.e.,
for every $\epsilon >0$ there is an
MD-strategy~$\zstrat_\epsilon$ such that for all $\state \in \states$ we have
$\probm_{\mdp,s,\zstrat_\epsilon}(\formula) \ge \valueof{\mdp}{s}-\epsilon$.
\end{theorem}

\newcommand{\lemreachopt}{
Let $\mdp=\mdptuple$ be an MDP, and $\formula = \reach{\reachset}$. % for some sink $\reachset \subseteq \states$.
Let $s_0 \in \states$ and $\zstrat$ be a strategy with $\probm_{\mdp,s_0,\zstrat}(\formula) = 1$.
Then there is an MD-strategy~$\hat\zstrat$ with $\probm_{\mdp,s_0,\hat\zstrat}(\formula) = 1$.
}

\begin{theorem}[follows from Proposition B in \cite{Ornstein:AMS1969}]\label{thm:reach-opt}
\lemreachopt
\end{theorem}

Both theorems are due to~\cite{Ornstein:AMS1969}; we give an alternative proof of Theorem~\ref{thm:reach-opt} in the appendix.
We generalize Theorem~\ref{thm:reach-opt} to B\"uchi objectives, using the principle that B\"uchi is repeated reachability:

\newcommand{\propasBuchi}{
Let $\mdp=\mdptuple$ be an MDP, and $\state_0 \in \states$, and $\zstrat$ a strategy, and $\coloring:\states\to \{1,2\}$, and $\formula=\Parity{\coloring}$.
Suppose $\probm_{\mdp,\state_0,\zstrat}(\formula) = 1$.
Then there is an MD-strategy~$\zstrat'$ with $\probm_{\mdp,\state_0,\zstrat'}(\formula) = 1$.
}

\begin{proposition}\label{prop:as-Buchi}
\propasBuchi
\end{proposition}

By appealing to Theorem~\ref{thm:reduction-to-as} it follows:

\newcommand{\thmBuchiopt}{
Let $\mdp$ be an MDP, $\coloring:\states\to \{1,2\}$, and
$\formula=\Parity{\coloring}$ be a B\"uchi-objective
(subsuming reachability and safety).
Then there exists an MD-strategy $\zstrat'$ that is optimal for all states
that have an optimal strategy:
\[
\begin{aligned}
\big(
\exists \zstrat \in \zstratset.\,
\probm_{\mdp,\state,\zstrat}(\formula) = \valueof{\mdp}{s}
\big)&
\Longrightarrow \\
&\probm_{\mdp,\state,\zstrat'}(\formula) = \valueof{\mdp}{s}
\end{aligned}
\]
}

\begin{theorem} \label{thm:Buchi-opt}
\thmBuchiopt
\end{theorem}

\section{When MD-strategies suffice in finitely branching MDPs}\label{sec-finite-branch}
In this section we prove that optimal strategies for $\cParity{\{0,1,2\}}$, where they exist,
can be chosen MD (Theorem~\ref{thm:012quant}) and that $\epsilon$-optimal strategies for co-B\"uchi objectives
can be chosen MD (Theorem~\ref{thm:coBuchi}).
To prepare the ground for these results, we first consider safety objectives.

\subsection{Optimal MD-strategies for Safety}
\label{subsec-safety}

The following proposition asserts in particular that for safety in finitely branching MDPs, there is no need for merely $\epsilon$-optimal strategies, as there always exists an optimal MD-strategy.

\newcommand{\propoptimalavoiding}{
Let $\mdp=\mdptuple$ be a finitely branching MDP, and $\reachset \subseteq \states$, and $\formula = \safety{\reachset}$.
Define an MD-strategy~$\optav$ (for ``optimal avoiding'') that, in each state~$\state$, picks a successor state with the largest value $\valueof{\mdp}{\state}=\sup_{\zstrat\in\zstratset}
\probm_{\mdp,\state,\zstrat}(\formula)$.
Then for all states $\state \in \states$ we have
$\probm_{\mdp,\state,\optav}(\formula) = \valueof{\mdp}{\state}$, i.e.,
$\optav$ is uniformly optimal.
}

\begin{proposition}[from Theorem 7.3.6(a) in \cite{Puterman:book}]\label{prop:optimal-avoiding}
\propoptimalavoiding
\end{proposition}
Note that, for infinitely branching MDPs, this definition of~$\optav$ would be
unsound, as ``the largest value'' might not exist.

% Proposition~\ref{prop:optimal-avoiding} follows from Theorem 7.3.6(a) in \cite{Puterman:book}; nevertheless, we give a simple proof in the appendix.

\begin{definition}\label{def:safeset}
Let $\mdp=\mdptuple$ be a finitely branching MDP,
$\coloring:\states\to \nat$ a color function, 
$\formula = \safety{\colorset \states \neq 0}$,
$\optav$ the strategy from Proposition~\ref{prop:optimal-avoiding}
and $\tau \in [0,1]$. We define
\[
 \safesub{\mdp}{\tau} := \{\state \in \states \mid 
\probm_{\mdp,\state,\optav}(\formula) \ge \tau\} \;,
\]
i.e., $\safesub{\mdp}{\tau}$ is the set of states from which the player can remain within color-0 states forever with probability $\ge \tau$.
We drop the subscript~$\mdp$ when the MDP~$\mdp$ is understood.
\end{definition}

Loosely speaking, the following lemma gives a lower bound on the probability that, starting from a ``safe'' state, ``unsafe'' states are forever avoided by~$\optav$:

\newcommand{\lemsafetwolevels}{
Let $\mdp=\mdptuple$ be a finitely branching MDP, $\coloring:\states\to \nat$
a color function and
$\optav$ the strategy from Proposition~\ref{prop:optimal-avoiding}.
Let $0 < \tau_1 \le \tau_2 \le 1$, and $\state \in \safe{\tau_2}$.
Then $\probm_{\mdp,\state,\optav}(\always \safe{\tau_1}) \ge \frac{\tau_2 - \tau_1}{1 - \tau_1}$.
}
\begin{lemma}\label{lem:safe-two-levels}
\lemsafetwolevels
\end{lemma}

\begin{proof}
We compute probabilities conditioned under the event $\always \safe{\tau_1}$.
Since $\safe{\tau_1} \subseteq \colorset \states = 0$, we have $\probm_{\mdp,\state,\optav}(\always \colorset \states = 0 \mid \always \safe{\tau_1}) = 1$.
From the definition of $\safe{\tau_1}$ and the Markov property we get $\probm_{\mdp,\state,\optav}(\always \colorset \states = 0 \mid \neg \always \safe{\tau_1}) \le \tau_1$.
Applying the law of total probability and writing $x$ for $\probm_{\mdp,\state,\optav}(\always \safe{\tau_1})$ we obtain:
\begin{align*}
\tau_2 
& \ \le \ \probm_{\mdp,\state,\optav}(\always \colorset \states = 0)  \qquad \qquad \qquad \qquad \text{Def.~\ref{def:safeset}} \\[1mm]
& \ = \ \probm_{\mdp,\state,\optav}(\always \colorset \states = 0 \mid \always \safe{\tau_1}) \cdot x  \\
& \hspace{4mm} \mbox{} + 
      \probm_{\mdp,\state,\optav}(\always \colorset \states = 0 \mid \neg \always \safe{\tau_1}) \cdot (1-x) \\[1mm]
& \ \le \ x + \tau_1 \cdot (1-x)
\end{align*}
It follows $x \ge \frac{\tau_2 - \tau_1}{1 - \tau_1}$.
\end{proof}

The following lemma states for all $\tau < 1$ that eventually remaining in color-0 states but outside~$\safe{\tau}$ has probability zero.

\newcommand{\lemaszotreturntosafe}{
Let $\mdp=\mdptuple$ be a finitely branching MDP, and $\coloring:\states\to \nat$ a color function.
Let $\state$ be a state, and $\zstrat$ a strategy, and $\tau < 1$.
Then $\probm_{\mdp,\state,\zstrat}(\eventually \always \neg\safe{\tau} \land \eventually \always \colorset \states = 0) = 0$.
}
\begin{lemma}\label{lem:as012-return-to-safe}
\lemaszotreturntosafe
\end{lemma}

\subsection{Optimal MD-strategies for $\cParity{\{0,1,2\}}$}\label{subsec-as012}
\newcommand{\thmzotquant}{
Let $\mdp$ be a finitely branching  MDP, $\coloring:\states\to \{0,1,2\}$, and
$\formula=\Parity{\coloring}$.
Then there exists an MD-strategy $\zstrat'$ that is optimal for all states
that have an optimal strategy:
\[
\begin{aligned}
\big(
\exists \zstrat \in \zstratset.\,
\probm_{\mdp,\state,\zstrat}(\formula) = \valueof{\mdp}{s}
\big)
&\Longrightarrow\\
&\probm_{\mdp,\state,\zstrat'}(\formula) = \valueof{\mdp}{s}
\end{aligned}
\]
}

\begin{theorem} \label{thm:012quant}
\thmzotquant
\end{theorem}

By appealing to Theorem~\ref{thm:reduction-to-as} it suffices to show:
\newcommand{\thmaszot}{
Let $\mdp=\mdptuple$ be a finitely branching MDP, and $\state_0 \in \states$, and $\zstrat$ a strategy, and $\coloring:\states\to \{0,1,2\}$, and $\formula=\Parity{\coloring}$.
Suppose $\probm_{\mdp,\state_0,\zstrat}(\formula) = 1$.
Then there is an MD-strategy~$\zstrat'$ with $\probm_{\mdp,\state_0,\zstrat'}(\formula) = 1$.
%In other words, MD-strategies suffice to win almost-sure $\cParity{\{0,1,2\}}$ objectives.
}
\begin{proposition}\label{prop:as012}
\thmaszot
\end{proposition}

\ignore{
For the proof we are going to use the following two lemmas from Puterman(??) about MD-strategies for co-reachability and reachability objectives, respectively.

\begin{lemma}[from Puterman ??]\label{lem:optimal-avoiding}
Let $\mdp=\mdptuple$ be an MDP, and $\coloring:\states\to \nat$ a color function.
Define an MD-strategy~$\hat \zstrat$ that, in each state~$\state$, picks a successor state with the largest value $\valueof{\mdp}{\state}=\sup_{\zstrat\in\zstratset}\probm_{\mdp,\state,\zstrat}(\always \colorset \states = 0)$. 
Then for all states $\state \in \states$ we have $\probm_{\mdp,\state,\hat \zstrat}(\always \colorset \states = 0) = \valueof{\mdp}{\state}$.
\end{lemma}

\begin{lemma}[from Puterman ??]\label{lem:MDP-reach}
Let $\mdp=\mdptuple$ be a finitely branching  MDP, and $\reachset \subseteq \states$.
There exists an MD-strategy~$\hat \zstrat$ such that for all $\state \in \states$:
\[
\big(
\exists \zstrat \in \zstratset.\,
\probm_{\mdp,\state,\zstrat}(\eventually \reachset) = 1
\big)
&\Longrightarrow\\
&\probm_{\mdp,\state,\hat\zstrat}(\eventually \reachset) = 1
\]
\end{lemma}
}

The following simple lemma provides a scheme for proving almost-sure properties.
\newcommand{\lemaspartitionscheme}{
Let $\probm$ be a probability measure over the sample space~$\Omega$.
Let $(\playset_i)_{i \in I}$ be a countable partition of~$\Omega$ in measurable events.
Let $E \subseteq \Omega$ be a measurable event.
Suppose $\probm(\playset_i \cap E) = \probm(\playset_i)$ holds for all $i \in I$.
Then $\probm(E) = 1$.
}

\begin{lemma} \label{lem-as-partition-scheme}
\lemaspartitionscheme
\end{lemma}

We are ready to prove Proposition~\ref{prop:as012}.

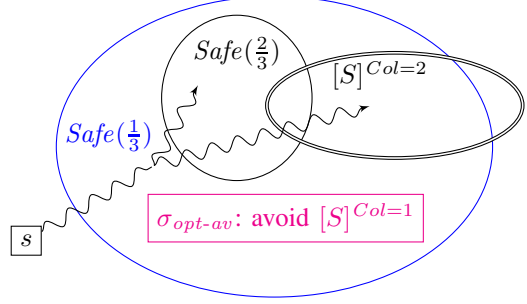
\begin{figure*}[t]
      \begin{center}
      			
      \centering
\begin{tikzpicture}[>=latex',shorten >=1pt,node distance=1.9cm,on grid,auto,
roundnode/.style={circle, draw,minimum size=1.5mm},
squarenode/.style={rectangle, draw,minimum size=2mm}]

\draw [blue](0,1.65) ellipse (2.9cm and 2cm);
\node [squarenode,draw=none,blue] (s) at (-2.2,1.8) {$\safe{\frac13}$};
\draw (-.5,2.3) ellipse (1cm and 1.1cm);
\node [squarenode,draw=none] (s) at (-.5,2.9) {$\safe{\frac23}$};
\draw [double] (1.6,2.2) ellipse (1.7cm and .7cm);
\node [squarenode,draw=none] (s) at (1.4,2.6) {$\colorset \states = 2$};

\node [squarenode,magenta] (s) at (0.15,.7) {$\optav$: avoid $\colorset \states = 1$};
\node [squarenode,magenta] (s) at (-4.5,4) {$\hat\zstrat$: almost-sure $\reach{\safe{\frac23} \cup \colorset \states = 2}$};
\node [squarenode] (s) at (-3.3,.4) {$ s $};
\draw [-,decorate,decoration=snake] (s) -- (-1.6,1.4);
\draw [->,decorate,decoration=snake] (-1.6,1.4)--(-1,2.5);
\draw [->,decorate,decoration=snake] (-1.6,1.4)--(1.3,2.2);

\end{tikzpicture} 
      
				\end{center}
		\caption{The almost-surely winning MD-strategy~$\zstrat'$ for $\cParity{\{0,1,2\}}$ is obtained by combining the MD-strategies~$\optav$
		and $\hat\zstrat$: play~$\optav$ inside~$\safe{\frac13}$ and $\hat\zstrat$ outside that set.
%All states are assumed to be almost-surely winning.
A key point is that fixing~$\optav$ inside~$\safe{\frac13}$ does not prevent~$\hat\zstrat$ from achieving its objective.
}
\label{fig:winningstrategy012parity}
\end{figure*}

\begin{proof}[Proof of Proposition~\ref{prop:as012}]
To achieve an almost-sure winning objective, the player must forever remain in states from which the objective can be achieved almost surely.
So we can assume without loss of generality that all states are almost-sure winning, i.e., for all $\state \in \states$ we have $\probm_{\mdp,\state,\zstrat}(\formula) = 1$ for some~$\zstrat$.

We will define an MD-strategy~$\zstrat'$ with $\probm_{\mdp,\state,\zstrat'}(\formula) = 1$ for all $\state \in \states$.
We first define the MD-strategy~$\zstrat'$ partially for the states in $\safesub{\mdp}{\frac13}$ and then extend the definition of~$\zstrat'$ to all states.
For the states in $\safesub{\mdp}{\frac13}$ define $\zstrat' := \optav$ as in Proposition~\ref{prop:optimal-avoiding}, see Figure~\ref{fig:winningstrategy012parity}.
Let $\mdp'$ be the MDP obtained from~$\mdp$ by restricting the transition relation as prescribed by the partial MD-strategy~$\zstrat'$.

For any $\tau \in [0,1]$, we have $\safesub{\mdp}{\tau} = \safesub{\mdp'}{\tau}$.
Indeed, since $\mdp'$ restricts the options of the player, we have $\safesub{\mdp}{\tau} \supseteq \safesub{\mdp'}{\tau}$.
Conversely, let $\state \in \safesub{\mdp}{\tau}$.
The strategy $\optav$ from Proposition~\ref{prop:optimal-avoiding} achieves $\probm_{\mdp,\state,\optav}(\always \colorset \states = 0) \ge \tau$.
Since $\optav$ can be applied in~$\mdp'$, and results in the same Markov chain
as applying it in $\mdp$, we conclude $\state \in \safesub{\mdp'}{\tau}$.
This justifies to write $\safe{\tau}$ for $\safesub{\mdp}{\tau} = \safesub{\mdp'}{\tau}$ in the remainder of the proof.

Next we show that, also in~$\mdp'$, for all states $\state \in \states$ there exists a strategy~$\zstrat_1$ with $\probm_{\mdp',\state,\zstrat_1}(\formula) = 1$.
This strategy~$\zstrat_1$ is defined as follows.
First play according to a strategy~$\zstrat$ from the statement of the theorem.
If and when the play visits $\safe{\frac13}$, switch to the MD-strategy~$\optav$ from Proposition~\ref{prop:optimal-avoiding}.
If and when the play then visits $\colorset \states \ne 0$, switch back to a strategy~$\zstrat$ from the statement of the theorem, and so forth.

We show that $\zstrat_1$ achieves $\probm_{\mdp',\state,\zstrat_1}(\formula) = 1$.
To this end we will use Lemma~\ref{lem-as-partition-scheme}.
We partition the runs of $s S^\omega$ in three events $\playset_0, \playset_1, \playset_2$ as follows:
\begin{itemize}
\item
$\playset_0$ contains the runs where $\zstrat_1$ switches between $\optav$ and~$\zstrat$ infinitely often.
\item
$\playset_1$ contains the runs where $\zstrat_1$ eventually only plays according to~$\optav$.
\item
$\playset_2$ contains the runs where $\zstrat_1$ eventually only plays according to~$\zstrat$.
\end{itemize}
Each time $\zstrat_1$ switches to~$\optav$, there is, by Proposition~\ref{prop:optimal-avoiding}, 
a probability of at least~$\frac13$ of never visiting a color-$\{1,2\}$ state again and thus of never again switching to~$\zstrat$.
It follows that $\probm_{\mdp',\state,\zstrat_1}(\playset_0) = 0$.
By the definition of $\optav$ we have $\playset_1 \subseteq \denotationof{\eventually \always \colorset \states = 0}{} \subseteq \denotationof{\formula}{}$, and hence $\probm_{\mdp',\state,\zstrat_1}(\playset_1 \cap \denotationof{\formula}{}) = \probm_{\mdp',\state,\zstrat_1}(\playset_1)$.
Since $\probm_{\mdp,\state,\zstrat}(\formula) = 1$ and $\formula$ is prefix-independent, we have
$\probm_{\mdp',\state,\zstrat_1}(\playset_2 \cap \denotationof{\formula}{}) = \probm_{\mdp',\state,\zstrat_1}(\playset_2)$.
Using Lemma~\ref{lem-as-partition-scheme}, we obtain $\probm_{\mdp',\state,\zstrat_1}(\formula) = 1$.

Next we show that for all $s \in \states$ the strategy~$\zstrat_1$ defined above achieves $\probm_{\mdp',\state,\zstrat_1}(\eventually \safe{\frac23} \lor \eventually \colorset \states = 2) = 1$.
To this end we will use Lemma~\ref{lem-as-partition-scheme} again.
We partition the runs of $s S^\omega$ into three events $\playset_1', \playset_2', \playset_0'$ as follows:
\begin{itemize}
\setlength\itemsep{.8em}
\item
$\playset_1' = \denotationof{\eventually \always \colorset \states = 0}{s}$
\item
$\playset_2' = \denotationof{\always \eventually \colorset \states = 2}{s}$
\item
$\playset_0' = s S^\omega \setminus \denotationof{\formula}{s}$
\end{itemize}
We have previously shown that $\probm_{\mdp',\state,\zstrat_1}(\formula) = 1$, hence $\probm_{\mdp',\state,\zstrat_1}(\playset_0') = 0$.
By Lemma~\ref{lem:as012-return-to-safe}, almost all runs in~$\playset_1'$ satisfy $\always \eventually \safe{\frac23}$.
Since $\denotationof{\always \eventually \safe{\frac23}}{} \subseteq \denotationof{\eventually \safe{\frac23}}{}$, we have $\probm_{\mdp',\state,\zstrat_1}(\playset_1' \cap \denotationof{\eventually \safe{\frac23} \lor \eventually \colorset \states = 2}{}) = \probm_{\mdp',\state,\zstrat_1}(\playset_1')$.
Since $\playset_2' \subseteq \denotationof{\eventually \colorset \states = 2}{}$, we also have $\probm_{\mdp',\state,\zstrat_1}(\playset_2' \cap \denotationof{\eventually \safe{\frac23} \lor \eventually \colorset \states = 2}{}) = \probm_{\mdp',\state,\zstrat_1}(\playset_2')$.
Using Lemma~\ref{lem-as-partition-scheme} we obtain $\probm_{\mdp',\state,\zstrat_1}(\eventually \safe{\frac23} \lor \eventually \colorset \states = 2) = 1$.

Writing $\reachset = \safe{\frac23} \cup \colorset \states = 2$ we have just shown that for all $s \in \states$ there is a strategy~$\zstrat_1$ with $\probm_{\mdp',\state,\zstrat_1}(\eventually \reachset) = 1$.
By Lemma~\ref{lem:MDP-as-uniform} there is an MD-strategy~$\hat\zstrat$ for~$\mdp'$ with $\probm_{\mdp',\state,\hat\zstrat}(\eventually \reachset) = 1$ for all $s \in \states$.
We extend the (so far partially defined) strategy~$\zstrat'$ by~$\hat\zstrat$.
Thus we obtain a (fully defined) strategy~$\zstrat'$ for~$\mdp$ such that for all $s \in \states$ we have $\probm_{\mdp,\state,\zstrat'}(\eventually \reachset) = 1$.

It remains to show that for all $s \in \states$ we have $\probm_{\mdp,\state,\zstrat'}(\formula) = 1$.
To this end we will use Lemma~\ref{lem-as-partition-scheme} again.
We partition the runs of $s S^\omega$ in two events $\playset_1'', \playset_2''$:
\begin{itemize}
\setlength\itemsep{.6em}
\item
$\playset_1'' = \denotationof{\always\eventually \safe{\frac23}}{s}$, i.e., $\playset_1''$ contains the runs that visit $\safe{\frac23}$ infinitely often.
\item
$\playset_2'' = \denotationof{\eventually\always \neg\safe{\frac23}}{s}$, i.e., $\playset_2''$ contains the runs that from some point on never visit $\safe{\frac23}$.
\end{itemize}
Every time a run enters $\safe{\frac23}$, by Lemma~\ref{lem:safe-two-levels}, the probability is at least $\frac12$ that the run remains in $\safe{\frac13}$ forever.
It follows that almost all runs in~$\playset_1''$ eventually remain in $\safe{\frac13}$ forever, i.e., $\probm_{\mdp,\state,\zstrat'}(\playset_1'' \cap \denotationof{\eventually \always \safe{\frac13}}{}) = \probm_{\mdp,\state,\zstrat'}(\playset_1'')$.
Since $\safe{\frac13} \subseteq \colorset \states = 0$, we have $\denotationof{\eventually \always \safe{\frac13}}{} \subseteq \denotationof{\eventually \always \colorset \states = 0}{} \subseteq \denotationof{\formula}{}$.
Hence also $\probm_{\mdp,\state,\zstrat'}(\playset_1'' \cap \denotationof{\formula}{}) = \probm_{\mdp,\state,\zstrat'}(\playset_1'')$.

We have previously shown that $\probm_{\mdp,\state,\zstrat'}(\eventually \reachset) = 1$ holds for all $\state \in \states$.
Hence also $\probm_{\mdp,\state,\zstrat'}(\always \eventually \reachset) = 1$ holds for all $\state \in \states$.
In particular, almost all runs in~$\playset_2''$ satisfy $\always \eventually \reachset$.
By comparing the definitions of $\playset_2''$ and $\reachset$ we see that almost all runs in~$\playset_2''$ even satisfy $\always \eventually \colorset \states = 2$.
Since $\denotationof{\always \eventually \colorset \states = 2}{} \subseteq \denotationof{\formula}{}$, we obtain $\probm_{\mdp,\state,\zstrat'}(\playset_2'' \cap \denotationof{\formula}{}) = \probm_{\mdp,\state,\zstrat'}(\playset_2'')$.

A final application of Lemma~\ref{lem-as-partition-scheme} yields
$\probm_{\mdp,\state,\zstrat'}(\formula) = 1$ for all $s \in \states$.
%, completing the proof.
\end{proof}

\subsection{$\epsilon$-Optimal MD-strategies for Co-B\"uchi}\label{subsec-cobuchi}

\newcommand{\thmcoBuchi}{
Let $\mdp=\mdptuple$ be a finitely branching MDP, 
$\coloring:\states\to \{0,1\}$, and $\formula=\Parity{\coloring}$ be the
co-B\"uchi objective.  Then there exist uniform $\epsilon$-optimal
MD-strategies.
I.e., for every $\epsilon >0$ there is an MD-strategy $\zstrat_\epsilon$ 
with $\probm_{\mdp,\state_0,\zstrat_\epsilon}(\formula) \ge
\valueof{\mdp}{\state_0} - \epsilon$
for every $\state_0 \in \states$.
}

\begin{theorem}\label{thm:coBuchi}
\thmcoBuchi
\end{theorem}
\begin{proof}
Let $\epsilon_1 >0$ be a suitably small number (to be determined later),
$\tau_1 := 1-\epsilon_1$ and
$\safesub{\mdp}{\tau_1}$ defined as in Definition~\ref{def:safeset}.
Let $\optav$ be the MD-strategy
from Proposition~\ref{prop:optimal-avoiding}.
From $\mdp$ we obtain a modified MDP $\mdp'$ by fixing all player choices from states in
$\safesub{\mdp}{\tau_1}$ according to $\optav$.

We show that $\valueof{\mdp'}{\state_0} \ge \valueof{\mdp}{\state_0} - \epsilon_1$.
By definition of the value $\valueof{\mdp}{\state_0}$, 
for every $\delta>0$ there exists a strategy $\zstrat_\delta$ in $\mdp$
s.t.\ $\probm_{\mdp,\state_0,\zstrat_\delta}(\formula) \ge \valueof{\mdp}{\state_0} -\delta$.
We define a strategy $\zstrat_\delta'$ in $\mdp'$ from state $\state_0$ as
follows.
First play like $\zstrat_\delta$. If and when a state in $\safesub{\mdp}{\tau_1}$ is
reached, play like $\optav$. This is possible, since no moves
from states outside $\safesub{\mdp}{\tau_1}$ have been fixed in $\mdp'$,
and all moves from states inside $\safesub{\mdp}{\tau_1}$ have been fixed 
according to $\optav$.
Then we have:
\[
\begin{aligned}
&\probm_{\mdp',\state_0,\zstrat_\delta'}(\formula) \\[1mm]
&=   \probm_{\mdp,\state_0,\zstrat_\delta}(\formula) \\
& 
\quad \;  - \, \probm_{\mdp,\state_0,\zstrat_\delta}(\eventually \safesub{\mdp}{\tau_1}) \cdot 
\probm_{\mdp,\state_0,\zstrat_\delta}(\formula | \eventually
\safesub{\mdp}{\tau_1}) \\                                           
&  \quad \; + \,\probm_{\mdp,\state_0,\zstrat_\delta}(\eventually
\safesub{\mdp}{\tau_1}) \cdot  \probm_{\mdp,\state_0,\zstrat_\delta'}(\formula | \eventually
\safesub{\mdp}{\tau_1}) \\[1mm]
& \ge   \probm_{\mdp,\state_0,\zstrat_\delta}(\formula) \\
&      
\quad \;- \,  \probm_{\mdp,\state_0,\zstrat_\delta}(\eventually \safesub{\mdp}{\tau_1}) \cdot 
\probm_{\mdp,\state_0,\zstrat_\delta}(\formula | \eventually
\safesub{\mdp}{\tau_1})  \\                                         
& \quad \; + \, \probm_{\mdp,\state_0,\zstrat_\delta}(\eventually
\safesub{\mdp}{\tau_1}) \cdot \tau_1 \\[1mm]
& \ge  \valueof{\mdp}{\state_0} - \delta - \probm_{\mdp,\state_0,\zstrat_\delta}(\eventually
\safesub{\mdp}{\tau_1})(1-\tau_1)\\[1mm]
& \ge  \valueof{\mdp}{\state_0} - \delta - \epsilon_1
\end{aligned}
\]
Since this holds for every $\delta >0$ we obtain 
$\valueof{\mdp'}{\state_0} \ge \valueof{\mdp}{\state_0} - \epsilon_1$.

Now let $\tau_2 := 1 - \epsilon_1/k$ for a suitably large $k \ge 1$ (to be
determined later) and
$\safesub{\mdp'}{\tau_2}$
be defined as in Definition~\ref{def:safeset}.
In particular, $\safesub{\mdp'}{\tau_2} = \safesub{\mdp}{\tau_2}$
(by the same argument as in the proof of Proposition~\ref{prop:as012}).

By definition of the value, for every $\epsilon_2 >0$ there exists a 
strategy $\zstrat_{\epsilon_2}$ in $\mdp'$ with
$\probm_{\mdp',\state_0,\zstrat_{\epsilon_2}}(\formula) \ge
\valueof{\mdp'}{\state_0} -\epsilon_2$.
Moreover, by Lemma~\ref{lem:as012-return-to-safe} and $\tau_2 <1$, 
$\probm_{\mdp',\state_0,\zstrat}(\eventually \safesub{\mdp'}{\tau_2}) \ge
\probm_{\mdp',\state_0,\zstrat}(\formula)$ for every strategy $\zstrat$ and
thus in particular for $\zstrat_{\epsilon_2}$.
Therefore, $\probm_{\mdp',\state_0,\zstrat_{\epsilon_2}}(\eventually \safesub{\mdp'}{\tau_2}) \ge
\valueof{\mdp'}{\state_0} -\epsilon_2$.
By Theorem~\ref{thm:reach-eps}, for every $\epsilon_3 >0$
there exists an MD-strategy $\zstrat'$
in $\mdp'$ with 
$\probm_{\mdp',\state_0,\zstrat'}(\eventually \safesub{\mdp'}{\tau_2}) \ge
\valueof{\mdp'}{\state_0} - \epsilon_2 - \epsilon_3$.
In particular, $\zstrat'$ must coincide with $\optav$
at all states in $\safesub{\mdp}{\tau_1}$, since in $\mdp'$ these choices are
already fixed.

We obtain the MD-strategy $\zstrat_\epsilon$ in $\mdp$ by 
combining the MD-strategies $\zstrat'$ and $\optav$.
The strategy $\zstrat_\epsilon$ plays like  
$\zstrat'$ at all states outside $\safesub{\mdp}{\tau_1}$ and like 
$\optav$ at all states inside $\safesub{\mdp}{\tau_1}$.

In order to show that $\zstrat_\epsilon$ has the required property
$\probm_{\mdp,\state_0,\zstrat_\epsilon}(\formula) \ge \valueof{\mdp}{\state_0} - \epsilon$, we 
first estimate the probability that a play according to $\zstrat_\epsilon$ will never
leave the set $\safesub{\mdp}{\tau_1}$ after having visited a state in
$\safesub{\mdp'}{\tau_2}$.

Let $\state \in \safesub{\mdp'}{\tau_2}$. 
Then, by Lemma~\ref{lem:safe-two-levels},
\[\begin{aligned}
\probm_{\mdp,\state,\optav}(\always \safe{\tau_1}) &\quad\ge \quad \frac{\tau_2 -
  \tau_1}{1 - \tau_1}\\[1mm] 
	& \quad= \quad \frac{(1-\epsilon_1/k) - (1-\epsilon_1)}{\epsilon_1}\\ 
	& \quad= \quad  1-\frac{1}{k}\,.
\end{aligned}\] 
In particular we also have 
$\probm_{\mdp,\state,\zstrat_\epsilon}(\always \safe{\tau_1}) \ge 
1-\frac{1}{k}$,
since $\zstrat_\epsilon$ coincides with $\optav$ inside the set
$\safesub{\mdp}{\tau_1}$.
Finally we obtain:
\begin{eqnarray*}
\probm_{\mdp,\state_0,\zstrat_\epsilon}(\formula)
& \ge & 
\probm_{\mdp,\state_0,\zstrat_\epsilon}(\eventually \safesub{\mdp'}{\tau_2})\\
&\quad&
{}\cdot
\probm_{\mdp,\state_0,\zstrat_\epsilon}(\eventually\always \safesub{\mdp}{\tau_1} |
\eventually \safesub{\mdp'}{\tau_2}) \\[1mm]
& \ge &
\probm_{\mdp',\state_0,\zstrat'}(\eventually \safesub{\mdp'}{\tau_2}) 
\cdot (1-1/k)
\\[1mm]
& \ge &
(\valueof{\mdp'}{\state_0} - \epsilon_2 - \epsilon_3)
\cdot
(1-1/k)\\[1mm]
& \ge & (\valueof{\mdp}{\state_0} - \epsilon_1 - \epsilon_2 - \epsilon_3) \cdot (1-1/k)
\end{eqnarray*}
This holds for every $\epsilon_1, \epsilon_2, \epsilon_3 >0$ and 
every $k \ge 1$, and moreover $\valueof{\mdp}{\state_0} \le 1$.
Thus we can set $\epsilon_1 = \epsilon_2 = \epsilon_3 := \epsilon/6$ and 
$k := \frac{2}{\epsilon}$ 
and obtain $\probm_{\mdp,\state_0,\zstrat_\epsilon}(\formula) 
\ge \valueof{\mdp}{\state_0} -\epsilon$ for every $\state_0 \in \states$ as required.
\end{proof}

%%\subsection{Applications in one-counter MDPs}\label{subsec-1c-mdp}
%
\section{Discussion}\label{sec-disc}
Our results on the memory requirements of ($\epsilon$)-optimal strategies 
(Figure~\ref{fig:results}) directly imply how much memory is needed to win
quantitative objectives of type $\zwinset{\formula}{\constraint\const}$
(considered, e.g., in \cite{BBKO:IC2011}).
For $\const <1$ the assumed winning strategy might have to be an $\epsilon$-optimal
one, since optimal strategies do not always exist. Thus MD-strategies are only
sufficient for reachability objectives in countable MDPs (resp., for $\cParity{\{0,1\}}$, safety
and reachability objectives in finitely branching MDPs).
In the special case of $\zwinset{\formula}{\ge 1}$ objectives (i.e., winning
almost-surely), the winning strategy (assuming it exists) must be optimal.
Thus MD-strategies are only sufficient for safety, reachability and
$\cParity{\{1,2\}}$ in countable MDPs 
(resp., for all objectives subsumed by $\cParity{\{0,1,2\}}$ in 
finitely branching MDPs).

In this paper we have studied countable MDPs. Not all our results carry over
to uncountable MDPs. The first issue is measurability. The probabilities are
only well-defined if the strategies are measurable functions, which might not 
exist without further conditions on the MDP; cf. Section 2.3
in \cite{Puterman:book}. Another issue is that strategies cannot 
generally be chosen \emph{uniform}, i.e., independent of the initial state.
E.g., in countable MDPs $\epsilon$-optimal strategies for reachability 
can be chosen uniform MD (Theorem~\ref{thm:reach-eps}), but this 
does not carry over to uncountable MDPs (Thm.~A
in \cite{Ornstein:AMS1969}). However, optimal strategies for reachability,
if they exist, can be chosen uniform MD (Proposition~B
in \cite{Ornstein:AMS1969}).

{\smallskip\bf\noindent Acknowledgements.}
This work was partially supported by the EPSRC
through grants EP/M027287/1, EP/M027651/1, EP/P020909/1 and EP/M003795/1
and by St. John's College, Oxford.

\bibliographystyle{plain}
\bibliography{base}
\newpage
\onecolumn
\appendix
\subsection{Proofs of Section~\ref{sec-parity}}
 
We  recall two results that are used throughout the proofs in this section:

\paragraph{Strong fairness of probabilistic choices in Markov chains.}
Given a Markov chain, let $\probm_q(\cdot)$ denotes 
the probability of events starting in  state~$q$ of the chain.
Let~$p,q$ be two states and~$\pi$ a finite path
starting in~$p$ with strictly positive probability.
Strong fairness  of probabilistic choices states that~$\probm_q(\always \eventually \pi)=\probm_q(\always \eventually p)$.
Intuitively, it means that under the condition that state~$p$ is visited infinitely often, 
any finite path starting in~$p$ will be taken infinitely often, almost surely~{\cite[Theorem~10.25]{BaierKatoenBook}}. 

\paragraph{The Borel-Cantelli lemma.} 
Suppose that $(E_n)_{n\in \nat}$ is a sequence of events in a probability space. 
Denote by $E^{\infty}$ the event 
$$\bigcap_{k=1}^{\infty} \bigcup_{n=k}^{\infty} E_n,$$
 that
intuitively is the event  ``$E_n$ occurs for infinitely many~$n$''. 
The Borel-Cantelli lemma asserts that if 
$\sum_{n=1}^{\infty} \probm(E_n)< \infty$ then $\probm(E^{\infty})=0$.
Informally speaking, if the sum of probabilities of the events~$E_n$ is bounded then
the probability that infinitely many of them occur is zero~\cite{billingsley-1995-probability}. 

\medskip
\begin{qtheorem}{\ref{thm:parity123}}
\thmparityonetwothree
\end{qtheorem}
\begin{proof}

Consider the MDP~$\mdp$ shown (on the left side) in Figure~\ref{fig:123MotwRequiresInfMemBuchiInfMem} where
$\rstates = \{r_i\}_{i \in\N}$ and
$\zstates=\{t\} \cup \{s_i\}_{i \in\N}$.
For all~$i \in \N$, there are transitions 
$s_i \transition r_i$ and $s_i \transition s_{i+1}$ in~$s_i$ states, whereas
$\probp(r_i,t)=2^{-i}$ and $\probp(r_i,s_0)=1-2^{-i}$ in random states.

\medskip

Let~$\sigma$ be an arbitrary FR-strategy. We prove that~$\probm_{\mdp,\state_0,\zstrat}(\formula)=0$.
By definition there is a transducer~$\memstratn$ with 
finite memory~$\memory$ and initial mode~$\memconf_0$ such that $\zstrat_{\memstratn}=\zstrat$.
Let $\mdp^{\memstratn}$ be the Markov chain obtained by the product of the MDP~$\mdp$ and the transducer~$\memstratn$.
The set of states in~$\mdp^{\memstratn}$ is~$\memory \times \states$. 
We define a coloring function for~$\mdp^{\memstratn}$ such that it ignores the memory mode and assigns to~$(\memconf,\state)$ 
 the same color as  state~$\state$ in~$\mdp$.
By a slight abuse of notation, we use the same notation~$\coloring$ for the coloring functions of both~$\mdp$ and~$\mdp^{\memstratn}$.
We also denote by~$\probm_{\mdp^{\memstratn},q}(\playset)$  the probability of a measurable set~$\playset$ of infinite paths (i.e., infinite plays),
starting in the state~$q$ of~$\mdp^{\memstratn}$. 

\medskip

We prove that $\probm_{\mdp,\state_0,\zstrat}(\always\eventually \colorset \states=2 \wedge \eventually\always \colorset \states\neq 3)=0$.
Equivalently, we show that $\probm_{\mdp^{\memstratn},(\memconf_0,s_0)}(\formula^{\memstratn})=0$ 
where $\formula^{\memstratn}=\always\eventually\colorset{\memory \times \states}{}{=2} \wedge \eventually\always \colorset{\memory \times \states}{}{\neq 3}$.
We proceed in three steps: 
we  will show, using  strong fairness of probabilistic choices
 in Markov chains, that for all  modes~$\memconf \in \memory$,
\begin{equation} \label{eq-str-fair}
\probm_{\mdp^{\memstratn},(\memconf_{0},s_0)}( \always \eventually (\memconf,s_0) \wedge \eventually\always \colorset{\memory \times \states}{}{\neq 3})=0.
\end{equation} 
Moreover,
\begin{equation} \label{eq-gf}
\denotationof{\always \eventually \colorset{\memory \times \states}{}{=2}}{} \subseteq \denotationof{\always \eventually (\memory \times \{s_0\})}{}
%\probm_{\mdp^{\memstratn},(\memconf_{0},s_0)} (\always \eventually (\memory \times \{s_0\}) \mid \always \eventually \colorset{\memory \times \states}{}{=2} )=1,
\end{equation} 
 
\noindent Since the memory of strategy~$\zstrat$ is finite ($\abs{\memory}<\infty$), we will show that

\begin{equation} \label{eq-last}
\denotationof{\always \eventually (\memory \times \{s_0\})}{}= \denotationof{\bigvee_{\memconf \in \memory}\always \eventually (\memconf,s_0)}{}
\end{equation} 

\noindent Using~(\ref{eq-str-fair})-(\ref{eq-last}) we complete the proof as follows:

\begin{align*}
\probm_{\mdp^{\memstratn},(\memconf_{0},s_0)}&( \always \eventually \colorset {\memory \times \states}{}{=2} \wedge \eventually\always \colorset{\memory \times \states}{}{\neq 3})\\
& \leq \probm_{\mdp^{\memstratn},(\memconf_{0},s_0)} (\bigvee_{\memconf \in \memory} (\always \eventually (\memconf,s_0) \wedge\eventually\always \colorset{\memory \times \states}{}{\neq 3})) &&\text{by (\ref{eq-gf}), (\ref{eq-last})}\\
& \leq \sum_{\memconf \in \memory} \probm_{\mdp^{\memstratn},(\memconf_{0},s_0)}( \always \eventually (\memconf,s_0) \wedge \eventually\always \colorset{\memory \times \states}{}{\neq 3})&& \text{union bound}\\
& =0.  &&\text{by~(\ref{eq-str-fair})}
\end{align*}

As a result $\probm_{\mdp,\state_0,\zstrat}(\formula)=0$. Below, we prove  (\ref{eq-str-fair}), (\ref{eq-gf}) and~(\ref{eq-last}).

\bigskip
We first highlight two properties of the Markov chain~$\mdp^{\memstratn}$. Consider the MDP~$\mdp$:
($i$) the only states with color~$2$ are random states~$r_i\in \rstates$, wherein  the only successors are~$s_0$ and~$t$.
	Hence,  from all states~$(\memconf,r_i)\in \memory\times \rstates$ in~$\mdp^{\memstratn}$,
	there are  successors~$q_1,q_2$ such that $q_1\in \memory\times \{s_0\}$ 
	and $q_2\in \memory\times \{t\}$;  moreover, all successors are in~$\memory\times \{s_0,t\}$. 
($ii$) The state $t$ has the unique successor~$s_0$.
		Hence, in~$\mdp^{\memstratn}$, from all states~$(\memconf,t)\in \memory\times \rstates$ 
	  all successors $q$ are such that~$q\in \memory\times \{s_0\}$.
		
\medskip

To establish~(\ref{eq-str-fair}), let~$(\memconf,s_0) \in \memory \times \{s_0\}$ be some state in the Markov chain~$\mdp^{\memstratn}$.
For the case~$\probm_{\mdp^{\memstratn},(\memconf_{0},s_0)} (\always \eventually (\memconf,s_0))=0$, we trivially have~(\ref{eq-str-fair}).
Therefore, we assume that~$\probm_{\mdp^{\memstratn},(\memconf_{0},s_0)} (\always \eventually (\memconf,s_0))>0$.
So,
there exists an infinite path  satisfying~$\always \eventually (\memconf,s_0)$.
Hence, there exists a finite path~$\pi$ from $(\memconf,s_0)$ to itself.
By the structure of the chain,  $\pi$
has a prefix $(\memconf_{s_0},s_0)(\memconf_{s_1},s_1) \cdots (\memconf_{s_i},s_i) \in (\memory \times \zstates)^*$ traversing player's states, 
with $\memconf_{s_0}=\memconf$, and next visiting some state~$(\memconf_{r_i},r_i)$.
By property~($i$), we know that~$(\memconf_{r_i},r_i)$ has a successor~$(\memconf_t,t)\in \memory\times \{t\}$. 
It implies that 
$$
\pi_t=(\memconf_{s_0},s_0)(\memconf_{s_1},s_1) \cdots (\memconf_{s_i},s_i)(\memconf_{r_i},r_i)
(\memconf_{t},t)
$$
is a finite path in $\mdp^{\memstratn}$, starting in~$(\memconf,s_0)$ with  positive probability. 
By strong fairness of probabilistic choices, 
$\probm_{\mdp^{\memstratn},(\memconf_{0},s_0)} (\always \eventually \pi_t)=\probm_{\mdp^{\memstratn},(\memconf_{0},s_0)}(\always \eventually (\memconf,s_0))$.
In other words,   under the condition that state~$(\memconf,s_0)$ is visited infinitely often, 
the finite path $\pi_t$ will be taken infinitely often, almost-surely.
As an immediate result of this and the fact that~$(\memconf_{t},t)$ has color~$3$, we conclude~(\ref{eq-str-fair}).

\medskip

To establish~(\ref{eq-gf}), 
we need to show that, for all infinite plays~$\pi^{\infty}$ of~$\mdp^{\memstratn}$,
if  $\pi^{\infty} \in \denotationof{\always \eventually \colorset{\memory \times \states}{}{=2}}{}$ 
then $\pi^{\infty}\in \denotationof{\always \eventually (\memory \times \{s_0\})}{}$. 
By the  properties~($i$) and~($ii$) of~$\mdp^{\memstratn}$,  whenever~$\pi^{\infty}$ visits some state from~$\memory\times \rstates$,
which are the only states with color~$2$,
then $\pi^{\infty}$ must visit some state from $\memory\times \{s_0\}$ within two steps. This results in~(\ref{eq-gf}).

\medskip

To establish~(\ref{eq-last}), we observe that
the inclusion $\supseteq$ is trivial. To show~$\subseteq$, let $\pi^{\infty}\in \denotationof{\always \eventually (\memory \times \{s_0\})}{}$ 
be an infinite path in the chain. As $\pi^{\infty}$ visits infinitely many elements from the finite set~$\memory \times \{s_0\}$,
 there must exist some element~$(\memconf ,s_0)$ that is visited infinitely often. Hence, 
$\pi^{\infty} \in \denotationof{\bigvee_{\memconf \in \memory} \always \eventually (\memconf,s_0)}{}$, which gives the inclusion and thus~(\ref{eq-last}). 

\bigskip
\bigskip

Now,  we construct an HD-strategy $\zstrat_h$ such that $\probm_{\mdp,\state_0,\zstrat_h}(\formula)=1$.
The strategy~$\zstrat_h$ is defined, for all partial plays~$\rho$, as follows:
$$
\zstrat_h(\rho)=
  \left\{ 
    \begin{array}{cl}
		  r_0~~~& \text{if } \rho=s_0\\
      r_k~~~& \text{if there exists } k>0 \text{ such that }\rho =(s_0 (\states\setminus\{s_0\})^{*})^{k-1}s_0 s_1 \cdots s_k\\
      s_j~~~&  \text{otherwise, where the last state visited by }\rho \text{ is } s_{j-1}.
    \end{array}
		\right.
$$
Intuitively, $\zstrat_h$ is such that
upon the $k$-th visit to state~$\state_0$,  
the path $\state_0 \state_1 \cdots \state_{k}$ is traversed and then the 
transition~$\state_{k} \transition r_{k}$ is chosen. 
Observe that $\probm_{\mdp,\state_0,\zstrat_h}(\always \eventually \colorset \states=2)=1$.
Below, we argue that $\probm_{\mdp,\state_0,\zstrat_h}(\always \eventually \colorset \states=3)=0$, 
which proves that~$\probm_{\mdp,\state_0,\zstrat_h}(\formula)=1$.

We define the sequence of events~$E_k$
of visiting~$t$ between the  $k$th and $k+1$st visits of~$s_0$.
For $k\geq 1$, let 
$$
E_k=(s_0 (\states\setminus\{s_0\})^{*})^{k-1}s_0 (\states\setminus\{s_0,t\})^* t s_0 S^{\omega}.
$$
Observe that 
$$
\bigcap_{n=1}^{\infty}\bigcup_{k\geq n}^{\infty} \pi(E_k)=\denotationof{\always \eventually \{t\}}{}=\denotationof{\always \eventually \colorset{\states}{}{=3}}{}.
$$
We  use the
Borel-Cantelli lemma to prove  that infinitely many of~$E_k$'s occur with zero probability, 
that is the probability of~$\always \eventually \colorset \states=3$. 
In fact, by construction of~$\zstrat_h$, observe that~$\probm_{\mdp,\state_0,\zstrat_h}(E_k)=2^{-k}$. 
Consequently, we have 
$$
\sum_{k=1}^{\infty}\probm_{\mdp,\state_0,\zstrat_h}(E_k)=1+\frac{1}{2}+\frac{1}{4}+\frac{1}{8}+\cdots=2< \infty.$$
By the  Borel-Cantelli lemma, we then have $\probm_{\mdp,\state_0,\zstrat_h}(\always \eventually \colorset \states=3)=0$, and thus
$\probm_{\mdp,\state_0,\zstrat_h}(\formula)=1$.
The proof is complete.
\end{proof}

\begin{qremark}{\ref{rem:parityonecounter}}
\remarkparityonecounter
\end{qremark}
\begin{proof}
Figure~\ref{fig:parity1counterMDP} shows a $1$-counter 
MDP with control-states~$\{s,r,r',t\}$ that 
is functionally equivalent to the one used 
in Theorem~\ref{thm:parity123} (Figure~\ref{fig:123MotwRequiresInfMemBuchiInfMem} (left)).
It just uses some auxiliary states that have no influence on the parity objective.
Starting in~$s$, the player can choose whether to increase the counter by~$1$
or to go to~$r$. 
In the random state~$r$ the behavior depends on the counter value.
If the counter is non-zero then the successors~$r,r'$ are chosen with equal 
probability and the counter is decreased by~$1$.
If the counter is zero then $t$ is the unique successor.
In state~$r'$ the counter is deterministically decreased until it becomes
zero, and then one goes to state~$s$.
The color function is 
$\coloring(s)=1$, $\coloring(r)=\coloring(r')=2$ and $\coloring(t)=3$.
If one is in state $r$ with counter value $n$, then the probability
of seeing state $t$ before returning to state $s$ is $2^{-n}$. 
 \begin{figure}[h]
      \begin{center}				
        \centering
\begin{tikzpicture}[>=latex',shorten >=1pt,node distance=1.9cm,on grid,auto,
roundnode/.style={circle, draw,minimum size=1.5mm},
squarenode/.style={rectangle, draw,minimum size=2mm},
diamonddnode/.style={diamond, draw,minimum size=2mm}]

\node [squarenode,initial,initial text={}] (s) at(0,0) [draw]{$~s~$};

\node[roundnode,double] (r)  [right=2cm of s] {$~r~$};
\node[roundnode,double] (rr)  [right=2cm of r] {$r'$};

\node [squarenode,double,inner sep = 4pt] (t)  [below left=1.2cmof r] {$ t $};
\node [rectangle, draw,inner sep = 2.4pt] (dum) [below left=1.2cmof r]{$ t $};

\path[->] (s) edge (r);
\path[->] (s) edge [loop above] node [above, midway] {$+1$} (s) ;
\path[->] (r) edge  node [above,midway] {$-1:\frac{1}{2}$} (rr);
\path[->] (r) edge [loop above]  node[above, midway] {$-1:\frac{1}{2}$} (r) ;
\path[->] (rr) edge [loop above]  node[above, midway] {$-1$} (rr);
\path[->, dashed] (r) edge node[right, near end] {$=0?$} (t);
\path[->, dashed] (rr) edge [bend left=110] node[below, midway] {$=0?$} (s);
\path[->] (t) edge  node[above, midway] {}(s);
\end{tikzpicture}
				\end{center}
		\caption{A $1$-counter MDP implementing an MDP similar to the
      one used in Theorem~\ref{thm:parity123}. 
		The dashed transitions (labeled with zero test~$=0?$) are taken when the counter is zero.}
\label{fig:parity1counterMDP}
\end{figure}
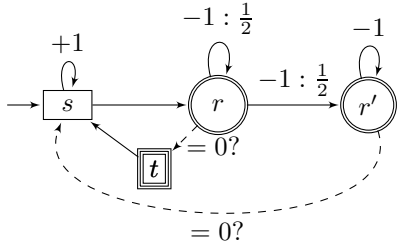
\end{proof}

\medskip
\begin{qtheorem}{\ref{thm:infBranchSafety}}
\thminfBranchSafety
\end{qtheorem}
\begin{proof}

Consider the MDP~$\mdp$ shown (on the left side) in Figure~\ref{fig:infBranchCoBuchiMaxSafety} where
$\zstates=\{s,t\}$ and $\rstates = \{r_i\}_{i \geq 1}$.
The state~$s$ is infinitely branching:  $s \transition r_i$ for  all~$i\geq 1$.
For all random states~$r_i$, there are two successors 
$\probp(r_i,t)=2^{-i}$ and $\probp(r_i,s_0)=1-2^{-i}$. The state~$t$ is a sink state.

\medskip
Let~$\sigma$ be an arbitrary FR-strategy. 
By definition there is a  transducer~$\memstratn$ with 
finite memory~$\memory$ and initial mode~$\memconf_0$ such that $\zstrat_{\memstratn}=\zstrat$.
Let $\mdp^{\memstratn}$ be the Markov chain obtained by the product of the MDP~$\mdp$ and the transducer~$\memstratn$.
The set of states in~$\mdp^{\memstratn}$ is~$\memory \times \states$. 
We  denote by~$\probm_{\mdp^{\memstratn},q}(\playset)$  the probability of a measurable set~$\playset$ of infinite paths (i.e., infinite plays),
starting in the state~$q$ of~$\mdp^{\memstratn}$. 

\medskip

Below, we  prove that~$\probm_{\mdp,\state,\zstrat}(\safety{\{t\}})=0$.
Equivalently, we show that $\probm_{\mdp^{\memstratn},(\memconf_0,s)}(\safety{\memory \times \{t\}})=0$. 
We proceed in the following three steps: We will show, using  strong fairness of probabilistic choices
 in Markov chains, that for all  modes~$\memconf \in \memory$,
\begin{equation} \label{eq-str-fair2}
\probm_{\mdp^{\memstratn},(\memconf_{0},s)}( \always \eventually (\memconf,s) \wedge \neg \eventually (\memory \times \{t\}))=0.
\end{equation} 
 
\noindent Since the memory of strategy~$\zstrat$ is finite ($\abs{\memory}<\infty$), we will show that

\begin{equation} \label{eq-last2}
\denotationof{\always \eventually (\memory \times \{s\})}{}= \denotationof{\bigvee_{\memconf \in \memory}\always \eventually (\memconf,s)}{}
\end{equation} 

\noindent Moreover, we  will prove that
\begin{equation} \label{eq-subset2}
\denotationof{\neg \eventually (\memory \times \{t\})}{}\subseteq \denotationof{\always \eventually  (\memory \times \{s\})}{}.
\end{equation} 

Using~(\ref{eq-str-fair2})-(\ref{eq-subset2}), we complete the proof as follows:
\begin{align*}
\probm_{\mdp^{\memstratn},(\memconf_0,s)}(&\neg \eventually (\memory \times \{t\}))\\
& = \probm_{\mdp^{\memstratn},(\memconf_{0},s)}( \always \eventually  (\memory \times \{s\}) \wedge \neg \eventually (\memory \times \{t\})) &&\text{by (\ref{eq-subset2})}\\
& = \probm_{\mdp^{\memstratn},(\memconf_{0},s)} (\bigvee_{\memconf \in \memory} (\always \eventually (\memconf,s_0) \wedge \neg \eventually (\memory \times \{t\})) &&\text{by  (\ref{eq-last2})}\\
& \leq \sum_{\memconf \in \memory} \probm_{\mdp^{\memstratn},(\memconf_{0},s)}( \always \eventually (\memconf,s_0) \wedge \neg \eventually (\memory \times \{t\}))&& \text{union bound}\\
& =0.  &&\text{by~(\ref{eq-str-fair2})}
\end{align*}

As a result, $\probm_{\mdp,\state_0,\zstrat}(\safety{\{t\}})=0$. Below we prove  (\ref{eq-str-fair2}), (\ref{eq-last2}) and~(\ref{eq-subset2}).

\bigskip

To establish~(\ref{eq-str-fair2}), let~$(\memconf,s) \in \memory \times \{s\}$ be some state in the Markov chain~$\mdp^{\memstratn}$.
For the case~$\probm_{\mdp^{\memstratn},(\memconf_{0},s)} (\always \eventually (\memconf,s))=0$, we trivially have~(\ref{eq-str-fair2}).
Therefore, we assume that~$\probm_{\mdp^{\memstratn},(\memconf_{0},s)} (\always \eventually (\memconf,s))>0$.
So,
there exists an infinite path  satisfying~$\always \eventually (\memconf,s)$.
Hence, there exists a finite path~$\pi$ from $(\memconf,s)$ to itself.
By the structure of the chain,  $\pi$ visits  some state~$(\memconf_{r_i},r_i)$.
In~$\mdp$, for all random states~$r_i\in \rstates$,  the only successors are~$s$ and~$t$.
Hence,  from all states~$(\memconf,r_i)\in \memory\times \rstates$ in~$\mdp^{\memstratn}$,
there is some  successor~$q\in \memory\times \{t\}$.
It implies that $\pi_t=(\memconf,s)(\memconf_{r_i},r_i)(\memconf_{t},t)$
is a finite path in $\mdp^{\memstratn}$, that starts in~$(\memconf,s)$ with  positive probability. 
By strong fairness of probabilistic choices, 
$\probm_{\mdp^{\memstratn},(\memconf_0,s)} (\eventually \pi_t)=\probm_{\mdp^{\memstratn},(\memconf_0,s)}(\always \eventually (\memconf,s))$.
As a result, we conclude~(\ref{eq-str-fair2}).

\medskip

To establish~(\ref{eq-last2}), we observe that
the inclusion $\supseteq$ is trivial. To show~$\subseteq$, let $\pi^{\infty}\in \denotationof{\always \eventually (\memory \times \{s_0\})}{}$ 
be an infinite path in the chain. As $\pi^{\infty}$ visits infinitely many elements from the finite set~$\memory \times \{s_0\}$,
 there must exist some element~$(\memconf ,s_0)$ that is visited infinitely often. Hence, 
$\pi^{\infty} \in \denotationof{\bigvee_{\memconf \in \memory} \always \eventually (\memconf,s_0)}{}$, which gives the inclusion and thus~(\ref{eq-last2}). 

\medskip
To establish~(\ref{eq-subset2}), 
note that in~$\mdp$, all successors of~$s$ are from random states, from which  the only successors are~$s$ and~$t$.
Hence,   for all infinite plays~$\pi^{\infty} \not \in \denotationof{\eventually (\memory \times \{t\})}{}$ in~$\mdp^{\memstratn}$,
$\pi^{\infty}$  alternates between some states from $\memory\times \{s\}$ and next some state from~$\memory\times \rstates$.
It implies that $\pi \in \denotationof{\always \eventually  (\memory \times \{s\})}{}$.

\bigskip
\bigskip

Now, let $c\in [0,1)$. Let~$n$ be such that $c\leq 1-\frac{1}{2^{n}}$. We construct an HD-strategy $\zstrat_n$ such that 
$$\probm_{\mdp,\state,\zstrat_n}(\formula)\geq 1-\frac{1}{2^{n}} \geq c.$$
For all partial plays~$\rho$, we define
$\zstrat_n(\rho)=r_{n+k}$ where $k$ is the number of times~$\rho$ has visited~$s$.
Intuitively, upon the $k$-th visit to~$s$, the strategy~$\zstrat_n$ chooses 
the transition~$s\transition r_{n+k}$.

For all $k\geq 1$, let $E_k$ be the event of visiting~$t$ after the $k$-th visit of~$s$, defined as follows
$$
E_k=(s \rstates)^{k-1}s \rstates t^{\omega}.
$$
Observe that $\denotationof{\eventually \{t\}}{}$ is the disjoint union of all $E_k$ events.
Hence, we have 
\begin{align*}
\probm_{\mdp,\state,\zstrat_n}&(\eventually \{t\})= \\
&\sum_{k=1}^{\infty} \probm_{\mdp,\state,\zstrat_n}(E_k)= &\text{by disjoint union}\\
&\sum_{k=1}^{\infty} \probm_{\mdp,\state,\zstrat_n} (\denotationof{(s \rstates)^{k-1}s \rstates}{}) \cdot \frac{1}{2^{n+k}}\\
%&\leq \sum_{k=1}^{\infty}  (1-\frac{1}{2^n})\times(1-\frac{1}{2^{n+1}})\times  \cdots \times (1-\frac{1}{2^{n+k-1}}) \times \frac{1}{2^{n+k}}\\
%&\leq \sum_{k=1}^{\infty}  (1-\frac{1}{2^{n}}) \times \frac{1}{2^{n+k}}\\
&\leq \sum_{k=1}^{\infty} \frac{1}{2^{n+k}}\\
%&\leq (1-\frac{1}{2^{n}}) \sum_{k=1}^{\infty} \frac{1}{2^{n+k}}=(1-\frac{1}{2^{n}})\cdot \frac{1}{2^{n}}\leq \frac{1}{2^{n}}.\\
&=\frac{1}{2^{n}}
\end{align*}
This proves that $\probm_{\mdp,\state,\zstrat_n}(\safety{\{t\}})= 1- \probm_{\mdp,\state,\zstrat_n}(\eventually \{t\})\geq 1-\frac{1}{2^{n}}\geq c$.
\end{proof}

\medskip
\begin{qtheorem}{\ref{thm:infBranchCoBuchi}}
\thminfBranchCoBuchi
\end{qtheorem}
\begin{proof}

Consider the MDP~$\mdp$ shown (on the right side) in Figure~\ref{fig:infBranchCoBuchiMaxSafety} where
$\zstates=\{s,t\}$ and $\rstates = \{r_i\}_{i \geq 1}$.
The state~$s$ is infinitely branching:  $s \transition r_i$ for  all~$i \geq 1$.
For all random states~$r_i$, there are two successors 
$\probp(r_i,t)=2^{-i}$ and $\probp(r_i,s)=1-2^{-i}$. The state~$t$ has the unique successor~$s$.

\medskip
Let~$\sigma$ be an arbitrary FR-strategy. 
By definition there is a  transducer~$\memstratn$ with 
finite memory~$\memory$ and initial mode~$\memconf_0$ such that $\zstrat_{\memstratn}=\zstrat$.
Let $\mdp^{\memstratn}$ be the Markov chain obtained by the product of the MDP~$\mdp$ and the transducer~$\memstratn$.
The set of states in~$\mdp^{\memstratn}$ is~$\memory \times \states$. 
We define a coloring function for~$\mdp^{\memstratn}$ such that it ignores the memory mode and assigns to~$(\memconf,\state)$ 
 the same color as  state~$\state$ in~$\mdp$. In particular, all states $q\in\memory \times \{t\}$ have color~$1$.
We use the same notation~$\coloring$ for the coloring functions of both~$\mdp$ and~$\mdp^{\memstratn}$.
We  denote by~$\probm_{\mdp^{\memstratn},q}(\playset)$  the probability of a measurable set~$\playset$ of infinite paths (i.e., infinite plays),
starting in the state~$q$ of~$\mdp^{\memstratn}$. 

\medskip

We prove that $\probm_{\mdp,\state,\zstrat}(\eventually\always \colorset \states\neq 1)=0$.
Equivalently, we show that $\probm_{\mdp^{\memstratn},(\memconf_0,s)}(\eventually\always \colorset{\memory \times \states}{}{\neq 1})=0$.
We proceed in three steps: 
we  will show, using  strong fairness of probabilistic choices
 in Markov chains, that for all  modes~$\memconf \in \memory$,
\begin{equation} \label{eq-str-fair3}
\probm_{\mdp^{\memstratn},(\memconf_{0},s)}( \always \eventually (\memconf,s) \wedge \eventually\always \colorset{\memory \times \states}{}{\neq 1})=0.
\end{equation}

\noindent Since the memory of strategy~$\zstrat$ is finite ($\abs{\memory}<\infty$), we will show that

\begin{equation} \label{eq-last3}
\denotationof{\always \eventually (\memory \times \{s\})}{}= \denotationof{\bigvee_{\memconf \in \memory}\always \eventually (\memconf,s)}{}
\end{equation} 

Moreover, we show that
\begin{equation} \label{eq-inclusion3}
\denotationof{\eventually \always \colorset{\memory \times \states}{}{\neq1}}{}\subseteq \denotationof{ \always \eventually (\memory \times \{s\})}{}
\end{equation}

\noindent Using~(\ref{eq-str-fair3}), (\ref{eq-last3}) and (\ref{eq-inclusion3}), we complete the proof as follows:
\begin{align*}
\probm_{\mdp^{\memstratn},(\memconf_{0},s)}&(\eventually\always \colorset{\memory \times \states}{}{\neq 1})=\\
&\probm_{\mdp^{\memstratn},(\memconf_{0},s)}( \always \eventually (\memory \times \{s\}) \wedge \eventually\always \colorset{\memory \times \states}{}{\neq 1})&&\text{by  (\ref{eq-inclusion3})}\\
& = \probm_{\mdp^{\memstratn},(\memconf_{0},s)} (\bigvee_{\memconf \in \memory} (\always \eventually (\memconf,s) \wedge\eventually\always \colorset{\memory \times \states}{}{\neq 1})) &&\text{by  (\ref{eq-last3})}\\
& \leq \sum_{\memconf \in \memory} \probm_{\mdp^{\memstratn},(\memconf_{0},s)}( \always \eventually (\memconf,s) \wedge \eventually\always \colorset{\memory \times \states}{}{\neq 1})&& \text{union bound}\\
& =0.  &&\text{by~(\ref{eq-str-fair3})}
\end{align*}

As a result, $\probm_{\mdp,\state,\zstrat}(\formula)=0$. Below we prove  (\ref{eq-str-fair3}), (\ref{eq-last3}) and~(\ref{eq-inclusion3}).
   
\medskip
	
To establish~(\ref{eq-str-fair3}), let~$(\memconf,s) \in \memory \times \{s\}$ be some state in the Markov chain~$\mdp^{\memstratn}$.
For the case~$\probm_{\mdp^{\memstratn},(\memconf_{0},s)} (\always \eventually (\memconf,s))=0$, we trivially have~(\ref{eq-str-fair3}).
Therefore, we assume that~$\probm_{\mdp^{\memstratn},(\memconf_{0},s)} (\always \eventually (\memconf,s))>0$.
So,
there exists an infinite path  satisfying~$\always \eventually (\memconf,s)$.
Hence, there exists a finite path~$\pi$ from $(\memconf,s)$ to itself.
By the structure of the chain,  $\pi$ visits  some state~$(\memconf_{r_i},r_i)$.
Consider the MDP~$\mdp$, for all random states~$r_i\in \rstates$,  the only successors are~$s$ and~$t$.
Hence,  from all states~$(\memconf,r_i)\in \memory\times \rstates$ in~$\mdp^{\memstratn}$,
there is some  successor~$q$ such that $q\in \memory\times \{t\}$.
It implies that $\pi_t=(\memconf,s)(\memconf_{r_i},r_i)(\memconf_{t},t)$
is a finite path in $\mdp^{\memstratn}$, that starts in~$(\memconf,s)$ with  positive probability. 
By strong fairness of probabilistic choices, 
$\probm_{\mdp^{\memstratn},(\memconf_0,s)} (\always\eventually \pi_t)=\probm_{\mdp^{\memstratn},(\memconf_0,s)}(\always \eventually (\memconf,s))$.
As a result, we conclude~(\ref{eq-str-fair3}).

\medskip

To establish~(\ref{eq-last3}), we observe that
the inclusion $\supseteq$ is trivial. To show~$\subseteq$, let $\pi^{\infty}\in \denotationof{\always \eventually (\memory \times \{s\})}{}$ 
be an infinite path in the chain. As $\pi^{\infty}$ visits infinitely many elements from the finite set~$\memory \times \{s\}$,
 there must exist some element~$(\memconf ,s)$ that is visited infinitely often. Hence, 
$\pi^{\infty} \in \denotationof{\bigvee_{\memconf \in \memory} \always \eventually (\memconf,s)}{}$, which gives the inclusion and thus~(\ref{eq-last3}). 

To establish~(\ref{eq-inclusion3}), 
note that in~$\mdp$, all infinite plays must visit $s$ infinitely often.
Thus all runs in $\mdp^{\memstratn}$ must visit $\memory \times \{\state\}$
infinitely often.
In particular, this holds for those infinite
runs~$\pi^{\infty} \in \denotationof{ \eventually\always \colorset{\memory \times \states}{}
{\neq 1}}{}$ in~$\mdp^{\memstratn}$.
Thus $\pi^{\infty} \in \denotationof{\always \eventually  (\memory \times \{s\})}{}$.

\bigskip

Now we construct an HD-strategy $\zstrat_h$ such that $\probm_{\mdp,\state,\zstrat_h}(\formula)=1$.
For all partial plays~$\rho$, we define
$\zstrat_h(\rho)=r_{k}$ where $k$ is the number of times~$\rho$ has visited~$s$.
Intuitively, upon the $k$-th visit to~$s$, the strategy~$\zstrat_h$ chooses the 
transition~$s\transition r_{k}$.
Below we argue that $\probm_{\mdp,\state,\zstrat_h}(\always \eventually \colorset \states=1)=0$, 
which proves that~$\probm_{\mdp,\state,\zstrat_h}(\formula)=1$.

We define the sequence of events~$E_k$
of visiting~$t$ between the $k$-th and $k+1$st visits of~$s$.
For $k\geq 1$, let 
$$
E_k=(s (\rstates+ \rstates t))^{k-1}s (\rstates t) s S^{\omega}.
$$
Observe that 
$$
\bigcap_{n=1}^{\infty}\bigcup_{k\geq n}^{\infty} \pi(E_k)=\denotationof{\always \eventually \{t\}}{}=\denotationof{\always \eventually \colorset{\states}{}{=1}}{}.
$$
We  use the
Borel-Cantelli lemma to prove  that infinitely many of~$E_k$'s occur with zero probability, 
that is the probability of~$\always \eventually \colorset \states=1$. 
In fact, by construction of~$\zstrat_h$, observe that~$\probm_{\mdp,\state,\zstrat_h}(E_k)=2^{-k}$. 
Consequently, we have 
$$
\sum_{k=1}^{\infty}\probm_{\mdp,\state_0,\zstrat_h}(E_k)=\frac{1}{2}+\frac{1}{4}+\frac{1}{8}+\cdots=1< \infty.$$
By the  Borel-Cantelli lemma, we then have $\probm_{\mdp,\state,\zstrat_h}(\always \eventually \colorset \states=1)=0$, and thus
$\probm_{\mdp,\state,\zstrat_h}(\formula)=1$.
The proof is complete.
\end{proof}

\subsection{Proofs of Section~\ref{sec-almost-to-optimal}}

\begin{lemma}\label{lem:prefix-ind-optimality}
\lemprefixindoptimality
\end{lemma}
\begin{proof}
First we show $\probm_{\mdp,s_0,\zstrat}(\denotationof{\formula}{s_0} \mid s_0 s_1 \cdots s_n \states^\omega) \le \valueof{\mdp}{s_n}$.
Define a strategy $\zstrat' : \states^*\zstates \to \dist(S)$ by $\zstrat'(w) = \zstrat(s_0 s_1 \cdots s_{n-1} w)$ for all $w \in \states^*\zstates$.
Then we have $\probm_{\mdp,s_0,\zstrat}(\denotationof{\formula}{s_0} \mid s_0 s_1 \cdots s_n \states^\omega) = \probm_{\mdp,s_n,\zstrat'}(\denotationof{\formula}{s_n}) \le \valueof{\mdp}{s_n}$.

Next we show $\valueof{\mdp}{s_n} \le \probm_{\mdp,s_0,\zstrat}(\denotationof{\formula}{s_0} \mid s_0 s_1 \cdots s_n \states^\omega)$.
Towards a contradiction, suppose that $\valueof{\mdp}{s_n} > \probm_{\mdp,s_0,\zstrat}(\denotationof{\formula}{s_0} \mid s_0 s_1 \cdots s_n \states^\omega)$.
Then, by the definition of $\valueof{\mdp}{s_n}$, there is a strategy~$\zstrat'$ with $\probm_{\mdp,s_n,\zstrat'}(\denotationof{\formula}{s_n}) > \probm_{\mdp,s_0,\zstrat}(\denotationof{\formula}{s_0} \mid s_0 s_1 \cdots s_n \states^\omega)$.
Define a strategy~$\zstrat''$ that plays according to~$\zstrat$;
if and when partial play $s_0 s_1 \cdots s_n$ is played, then $\zstrat''$ acts like~$\zstrat'$ henceforth;
otherwise $\zstrat''$ continues with~$\zstrat$ forever.
Using prefix-independence we get:
\begin{align*}
& \probm_{\mdp,s_0,\zstrat''}(\denotationof{\formula}{s_0}) \\
& = \probm_{\mdp,s_0,\zstrat''}(\denotationof{\formula}{s_0} \mid s_0 s_1 \cdots s_n \states^\omega) \cdot \probm_{\mdp,s_0,\zstrat''}(s_0 s_1 \cdots s_n \states^\omega) \\
& \ + \probm_{\mdp,s_0,\zstrat''}(\denotationof{\formula}{s_0} \setminus s_0 s_1 \cdots s_n \states^\omega) \\
& = \probm_{\mdp,s_n,\zstrat'}(\denotationof{\formula}{s_n}) \cdot \probm_{\mdp,s_0,\zstrat}(s_0 s_1 \cdots s_n \states^\omega) \\
& \ + \probm_{\mdp,s_0,\zstrat}(\denotationof{\formula}{s_0} \setminus s_0 s_1 \cdots s_n \states^\omega) && \text{def.~of~$\zstrat''$}\\
& > \probm_{\mdp,s_0,\zstrat}(\denotationof{\formula}{s_0} \mid s_0 s_1 \cdots s_n \states^\omega) \cdot \probm_{\mdp,s_0,\zstrat}(s_0 s_1 \cdots s_n \states^\omega) \\
& \ + \probm_{\mdp,s_0,\zstrat}(\denotationof{\formula}{s_0} \setminus s_0 s_1 \cdots s_n \states^\omega) && \text{def.~of~$\zstrat'$}\\
& = \probm_{\mdp,s_0,\zstrat}(\denotationof{\formula}{s_0}) \\
& = \valueof{\mdp}{s_0} && \text{def.~of~$\zstrat$}
\end{align*}
This contradicts the definition of $\valueof{\mdp}{s_0}$.
Hence we have shown item~1.

Towards items 2~and~3, we extend $\zstrat:\states^*\zstates \to \dist(\states)$ to $\zstrat : \states^* \states \to \dist(\states)$ by defining $\zstrat(w s) = \probp(s)$ for $w \in \states^*$ and $s \in \rstates$.
Then we have for all $s_{n+1} \in \states$:
\begin{equation} \label{eq:lem:prefix-ind-optimality}
\probm_{\mdp,s_0,\zstrat}(s_0 s_1 \cdots s_n s_{n+1} \states^\omega)
= \probm_{\mdp,s_0,\zstrat}(s_0 s_1 \cdots s_n \states^\omega)
  \cdot \zstrat(s_0 s_1 \cdots s_n)(s_{n+1})
\end{equation}
Further we have:
\begin{align*}
& \valueof{\mdp}{s_n} \\
& = \probm_{\mdp,s_0,\zstrat}(\denotationof{\formula}{s_0} \mid s_0 s_1 \cdots s_n \states^\omega)
  && \text{by item 1} \\
& = \frac{\probm_{\mdp,s_0,\zstrat}(\denotationof{\formula}{s_0} \cap s_0 s_1 \cdots s_n \states^\omega)}{\probm_{\mdp,s_0,\zstrat}(s_0 s_1 \cdots s_n \states^\omega)} \\
& = \frac{\sum_{s_{n+1} \in \states} \probm_{\mdp,s_0,\zstrat}(\denotationof{\formula}{s_0} \cap s_0 s_1 \cdots s_n s_{n+1} \states^\omega)}{\probm_{\mdp,s_0,\zstrat}(s_0 s_1 \cdots s_n \states^\omega)} \\
& = \frac{1}{\probm_{\mdp,s_0,\zstrat}(s_0 s_1 \cdots s_n \states^\omega)}
\cdot \sum_{s_{n+1} \in \states} \probm_{\mdp,s_0,\zstrat}(s_0 s_1 \cdots s_n s_{n+1} \states^\omega)
\cdot \mbox{} \\
& \hspace{47mm} \mbox{} \cdot \probm_{\mdp,s_0,\zstrat}(\denotationof{\formula}{s_0} \mid s_0 s_1 \cdots s_n s_{n+1} \states^\omega) \\
& = \sum_{s_{n+1} \in \states}  \zstrat(s_0 s_1 \cdots s_n)(s_{n+1})
\cdot \probm_{\mdp,s_0,\zstrat}(\denotationof{\formula}{s_0} \mid s_0 s_1 \cdots s_n s_{n+1} \states^\omega)
 && \text{by~\eqref{eq:lem:prefix-ind-optimality}} \\
& = \sum_{s_{n+1} \in \states}  \zstrat(s_0 s_1 \cdots s_n)(s_{n+1})
\cdot \valueof{\mdp}{s_{n+1}} 
  && \text{by item 1}
\end{align*}
Thus we have shown item~2.
Towards item~3, suppose $s_n \in \zstates$.
Then prefix-independence implies $\valueof{\mdp}{s_n} \ge \valueof{\mdp}{s_{n+1}}$ for all $s_{n+1}$ with $s_n \transition s_{n+1}$.
Since $\zstrat(s_0 s_1 \cdots s_n)$ is a probability distribution, the equality chain above shows that $\valueof{\mdp}{s_n} = \valueof{\mdp}{s_{n+1}}$ for all $s_{n+1} \in \supp(\sigma(s_0 s_1 \cdots s_n))$.
Thus we have shown item 3.
\end{proof}

\begin{qlemma}{\ref{lem:conditioned-construction}}
\lemconditionedconstruction
\end{qlemma}

\medskip
\begin{proof}
Note that by Lemma~\ref{lem:prefix-ind-optimality}.2 we have that $\pprobp(s)$ is a probability distribution for all $s \in \prstates$; hence the MDP~$\pmdp$ is well-defined.

We prove item~1 by induction on~$n$.
For $n=0$ it is trivial.
For the step, suppose that the equality in item~1 holds for some~$n$.
If $s_n \in \prstates$ then we have:
\begin{align*}
& \probm_{\pmdp,s_0,\zstrat}(s_0 s_1 \cdots s_n s_{n+1} \states^\omega) \\
& = \probm_{\pmdp,s_0,\zstrat}(s_0 s_1 \cdots s_n \states^\omega) \cdot \pprobp(s_n)(s_{n+1}) \\
& = \probm_{\mdp,s_0,\zstrat}(s_0 s_1 \cdots s_n \states^\omega) \cdot \frac{\valueof{\mdp}{s_n}}{\valueof{\mdp}{s_0}} \cdot \pprobp(s_n)(s_{n+1}) && \text{ind.\ hyp.} \\
& = \probm_{\mdp,s_0,\zstrat}(s_0 s_1 \cdots s_n \states^\omega) \cdot \frac{\valueof{\mdp}{s_n}}{\valueof{\mdp}{s_0}} \cdot \probp(s_n)(s_{n+1}) \cdot \frac{\valueof{\mdp}{s_{n+1}}}{\valueof{\mdp}{s_n}} && \text{def.~of~$\pprobp$} \\
& = \probm_{\mdp,s_0,\zstrat}(s_0 s_1 \cdots s_n s_{n+1} \states^\omega) \cdot \frac{\valueof{\mdp}{s_{n+1}}}{\valueof{\mdp}{s_0}}
\end{align*}
Let now $s_n \in \pzstates$.
If $\zstrat(s_0 s_1 \ldots s_n)(s_{n+1}) = 0$ then the inductive step is trivial.
Otherwise we have:
\begin{align*}
& \probm_{\pmdp,s_0,\zstrat}(s_0 s_1 \cdots s_n s_{n+1} \states^\omega) \\
& = \probm_{\pmdp,s_0,\zstrat}(s_0 s_1 \cdots s_n \states^\omega) \cdot \zstrat(s_0 s_1 \ldots s_n)(s_{n+1}) \\
& = \probm_{\mdp,s_0,\zstrat}(s_0 s_1 \cdots s_n \states^\omega) \cdot \frac{\valueof{\mdp}{s_n}}{\valueof{\mdp}{s_0}} \cdot \zstrat(s_0 s_1 \ldots s_n)(s_{n+1}) && \text{ind.\ hyp.} \\
& = \probm_{\mdp,s_0,\zstrat}(s_0 s_1 \cdots s_n \states^\omega) \cdot \frac{\valueof{\mdp}{s_{n+1}}}{\valueof{\mdp}{s_0}} \cdot \zstrat(s_0 s_1 \ldots s_n)(s_{n+1}) && \text{def.~of~$\mathord{\ptransition}$} \\
& = \probm_{\mdp,s_0,\zstrat}(s_0 s_1 \cdots s_n s_{n+1} \states^\omega) \cdot \frac{\valueof{\mdp}{s_{n+1}}}{\valueof{\mdp}{s_0}}
\end{align*}
This completes the inductive step, and we have proved item~1.

Towards item~2, let $s_0 \in \pstates$ and $\zstrat \in \zstratset_{\mdp}$ such that $\probm_{\mdp,s_0,\zstrat}(\formula) = \valueof{\mdp}{s_0} > 0$.
Observe that $\zstrat$ can be applied also in the MDP~$\pmdp$.
Indeed, for any $s \in \pzstates$, if $t$ is a possible successor state of~$s$ under~$\zstrat$, then $\valueof{\mdp}{s} = \valueof{\mdp}{t}$ by Lemma~\ref{lem:prefix-ind-optimality}.3 and thus $t \in \pstates$.

Let again $n \ge 0$ and $s_0, s_1, \ldots, s_n \in \states$.
\begin{itemize}
\item
Suppose $s_0 s_1 \cdots s_n$ is a partial play in~$\pmdp$ induced by~$\zstrat$.
Then we have:
\begin{align*}
& \probm_{\pmdp,s_0,\zstrat}(s_0 s_1 \cdots s_n \states^\omega) \cdot \probm_{\mdp,s_0,\zstrat}(\formula) \\
& = \probm_{\mdp,s_0,\zstrat}(s_0 s_1 \cdots s_n \states^\omega) \cdot \frac{\valueof{\mdp}{s_n}}{\valueof{\mdp}{s_0}}
\cdot \probm_{\mdp,s_0,\zstrat}(\formula)
 && \text{item~1} \\
& = \probm_{\mdp,s_0,\zstrat}(s_0 s_1 \cdots s_n \states^\omega) \cdot \valueof{\mdp}{s_n}
 && \text{assumption~on~$\zstrat$} \\
& = \probm_{\mdp,s_0,\zstrat}(s_0 s_1 \cdots s_n \states^\omega) \cdot \probm_{\mdp,s_0,\zstrat}(\denotationof{\formula}{s_0} \mid s_0 s_1 \cdots s_n \states^\omega)
 && \text{Lemma~\ref{lem:prefix-ind-optimality}.1} \\
& = \probm_{\mdp,s_0,\zstrat}(\denotationof{\formula}{s_0} \cap s_0 s_1 \cdots s_n \states^\omega)
%& = \probm_{\mdp,s_0,\zstrat}(\denotationof{\formula}{s_0} \cap s_0 s_1 \cdots s_n \pstates^\omega)
% && \text{$\probm_{\mdp,s_0,\zstrat}(\formula \land \eventually (\states \setminus \pstates)) = 0$\;,}
\end{align*}
%where the last equality follows from the fact that all states~$s$ reachable under the optimal strategy~$\zstrat$ must satisfy $s \in \pstates$ or $\valueof{\mdp}{s} = 0$.
\item
Suppose $s_0 s_1 \cdots s_n$ is not a partial play in~$\pmdp$ induced by~$\zstrat$.
Hence $\probm_{\pmdp,s_0,\zstrat}(s_0 s_1 \cdots s_n \states^\omega) = 0$.
If $s_0 s_1 \cdots s_n$ is not a partial play in~$\mdp$ induced by~$\zstrat$ then $\probm_{\mdp,s_0,\zstrat}(s_0 s_1 \cdots s_n \states^\omega) = 0$.
Otherwise, since $\zstrat$ is optimal, there is $i \le n$ with $\valueof{\mdp}{s_i} = 0$, hence $\probm_{\mdp,s_0,\zstrat}(\denotationof{\formula}{s_0} \cap s_0 s_1 \cdots s_n \states^\omega)$.
In either case we have
$
\probm_{\pmdp,s_0,\zstrat}(s_0 s_1 \cdots s_n \states^\omega) \cdot \probm_{\mdp,s_0,\zstrat}(\formula)
= 0 
= \probm_{\mdp,s_0,\zstrat}(\denotationof{\formula}{s_0} \cap s_0 s_1 \cdots s_n \states^\omega)
$.
\end{itemize}
In either case we have the equality $\probm_{\pmdp,s_0,\zstrat}(\playset) = \probm_{\mdp,s_0,\zstrat}(\playset \mid \denotationof{\formula}{s_0})$ for cylinders $\playset = s_0 s_1 \cdots s_n \states^\omega$.
Since probability measures extend uniquely from cylinders~\cite{billingsley-1995-probability}, the equality holds for all measurable $\playset \subseteq s_0 \states^\omega$.
Thus we have shown item~2.
\end{proof}

%\begin{qlemma}{\ref{lem:MDP-as-uniform}}
%\lemMDPasuniform
%\end{qlemma}

\begin{qlemma} {\ref{lem:measure-theory}}
\lemmeasuretheory
\end{qlemma}
\begin{proof}
Let $\classcyl = \{\cyl \subseteq s \states^\omega \mid \cyl \text{ cylinder}\}$ denote the class of cylinders.
This class generates an algebra $\classcyl_* \supseteq \classcyl$, which is the closure of~$\classcyl$ under finite union and complement.
The classes $\classcyl$ and $\classcyl_*$ generate the same $\sigma$-algebra $\sigma(\classcyl)$.
The class~$\classcyl_*$ is the set of finite disjoint unions of cylinders~\cite[Section~2]{billingsley-1995-probability}.
Hence $x \cdot \probm(\playset) \le \probm'(\playset)$ for all $\playset \in \classcyl_*$.

Define
\[
 \classmon = \{\playset \in \sigma(\classcyl) \mid x \cdot \probm(\playset) \le \probm'(\playset) \}\,.
\]
We have $\classcyl \subseteq \classcyl_* \subseteq \classmon \subseteq \sigma(\classcyl)$.
We show that $\classmon$ is a \emph{monotone} class, i.e., if $\playset_1, \playset_2, \ldots \in \classmon$, then $\playset_1 \subseteq \playset_2 \subseteq \cdots$ implies $\bigcup_i \playset_i \in \classmon$, and $\playset_1 \supseteq \playset_2 \supseteq \cdots$ implies $\bigcap_i \playset_i \in \classmon$.
Suppose $\playset_1, \playset_2, \ldots \in \classmon$ and $\playset_1 \subseteq \playset_2 \subseteq \cdots$.
Then:
\begin{align*}
x \cdot \probm\Big(\bigcup_i \playset_i\Big) 
& = \sup_i x \cdot \probm(\playset_i) && \text{measures are continuous from below} \\
& \le \sup_i \probm'(\playset_i) && \text{definition of~$\classmon$} \\
& = \probm'\Big(\bigcup_i \playset_i\Big) && \text{measures are continuous from below}
\end{align*}
So $\bigcup_i \playset_i \in \classmon$.
Using the fact that measures are continuous from above, one can similarly show that if $\playset_1, \playset_2, \ldots \in \classmon$ and $\playset_1 \supseteq \playset_2 \supseteq \cdots$ then $\bigcap_i \playset_i \in \classmon$.
Hence $\classmon$ is a monotone class.

Now the \emph{monotone class theorem} (see, e.g., \cite[Theorem~3.4]{billingsley-1995-probability}) implies that $\sigma(\classcyl) \subseteq \classmon$, thus $\classmon = \sigma(\classcyl)$.
Hence $x \cdot \probm(\playset) \le \probm'(\playset)$ for all $\playset \in \sigma(\classcyl)$.
\end{proof}

\subsection{Proofs of Section~\ref{sec-inf-MDPs}}

\begin{qtheorem}{\ref{thm:reach-opt}}
\lemreachopt
\end{qtheorem}
\begin{proof}
We can assume that $\reachset = \{t\}$ for some $t \in \states$.
We can also assume that all states are almost-surely winning, since in order to achieve an almost-sure winning objective, the player must forever remain in almost-surely winning states.

Let $\epsilon_1 := 1/2$.
By Theorem~\ref{thm:reach-eps} there exists an MD-strategy~$\zstrat_1$ such that $\probm_{\mdp,s_0,\zstrat_1}(\formula) \ge 1 - \epsilon_1$.
In fact, by the proof of Theorem~\ref{thm:reach-eps} there exists a finite subset $V_1 \subseteq \states$ such that $\probm_{\mdp,s_0,\zstrat_1}(\denotationof{\formula}{s_0} \cap V_1^\omega) \ge 1 - \epsilon_1$.
Let $U_1$ denote the states that occur in those plays that are both contained in $\denotationof{\formula}{s_0} \cap V_1^\omega$ and induced by~$\zstrat_1$.
Then $\probm_{\mdp,s_0,\zstrat_1}(\denotationof{\formula}{s_0} \cap U_1^\omega) = \probm_{\mdp,s_0,\zstrat_1}(\denotationof{\formula}{s_0} \cap V_1^\omega) \ge 1 - \epsilon_1$.
By the definition of~$U_1$, for all $s \in U_1$ the MD-strategy~$\zstrat_1$ induces a play from~$s_0$ to~$t$ via~$s$.
Hence we have $\probm_{\mdp,s,\zstrat_1}(\denotationof{\formula}{s} \cap U_1^\omega) > 0$.
Since $U_1 \subseteq V_1$ is finite, there are $c>0$ and $\ell \in \nat$ such that for all $s \in U_1$ we have $\probm_{\mdp,s,\zstrat_1}(s U_1^{\le\ell-1} \{t\}^\omega) \ge c$, i.e., from any state in~$U_1$ the probability that $t$ is reached in $\le\ell$~steps is at least~$c$.

Consider the MDP~$\mdp_1$ obtained from~$\mdp$ by fixing~$\zstrat_1$ on~$U_1$ (i.e., in~$\mdp_1$ the states in~$U_1$ are random states).
We argue that all states are almost-surely winning in~$\mdp_1$.
Indeed, let $s \in \states$ be any state.
Recall that $s$ is almost-surely winning in~$\mdp$.
Define an HR-strategy~$\zstrat$ in~$\mdp_1$ as follows:
first play a strategy that is almost-surely winning for~$s$ in~$\mdp$;
if and when $U_1$ is entered and then left again (entering some state $s' \in \states \setminus U_1$) then forget the history and play again a strategy that is almost-surely winning for~$s'$ in~$\mdp$; and so forth.
This strategy~$\zstrat$ reaches $\{t\}$ with probability~$1$ whenever the play
stays outside of~$U_1$.
I.e., almost all plays that eventually always avoid $U_1$ reach $\{t\}$.
Moreover, whenever the play enters~$U_1$, the probability that $t$ is reached
in $\le\ell$~steps is at least~$c$, i.e., there is a uniform bound.
Thus almost all plays that enter $U_1$ infinitely often reach $\{t\}$.
It follows that we have $\probm_{\mdp_1,\state,\zstrat}(\formula) = 1$.

Now we repeat the argument, but with $\mdp_1$ instead of~$\mdp$ and with $\epsilon_2 = 1/4$ instead of~$\epsilon_1$.
This yields a set $U_2 \supseteq U_1$ and an MD-strategy~$\zstrat_2$ that agrees with~$\zstrat_1$ on~$U_1$ so that $\probm_{\mdp,s_0,\zstrat_2}(\denotationof{\formula}{s_0} \cap U_2^\omega) = \probm_{\mdp_1,s_0,\zstrat_2}(\denotationof{\formula}{s_0} \cap U_2^\omega) = 1 - \epsilon_2$.
Similarly as before, obtain an MDP~$\mdp_2$ from~$\mdp_1$ by fixing $\zstrat_2$ on~$U_2$.
Then repeat again, and so forth, with $\epsilon_i = 1/2^i$ for $i = 1, 2, \ldots$

Define $U := \bigcup_{i \ge 1} U_i$.
Observe that on all $s \in U$ almost all (i.e., all except finitely many) strategies~$\zstrat_i$ agree.
Let $\hat \zstrat$ be an MD-strategy that on all states in~$U$ agrees with almost all MD-strategies~$\zstrat_i$.
By our construction we have $\probm_{\mdp,s_0,\hat\zstrat}(\denotationof{\formula}{s_0} \cap U^\omega) \ge 1 - \epsilon$ for all $\epsilon > 0$.
Hence $\probm_{\mdp,s_0,\hat\zstrat}(\denotationof{\formula}{s_0} \cap U^\omega) = 1$.
\end{proof}

\medskip

\begin{qproposition}{\ref{prop:as-Buchi}}
\propasBuchi
\end{qproposition}

\begin{proof}
We can assume that all states are almost-surely winning, since in order to achieve an almost-sure winning objective, the player must forever remain in almost-surely winning states.
We provide an MD-strategy~$\hat\zstrat$ such that for all states $s \in \states$ we have $\probm_{\mdp,\state,\hat\zstrat}(\formula)=1$.

Set $\formula' = \reach{\colorset \states = 2}$.
Note that $\denotationof{\formula}{} \subseteq \denotationof{\formula'}{}$.
Since all states are almost-surely winning for~$\formula$, all states are almost-surely winning for~$\formula'$.
By Theorem~\ref{thm:reach-opt} and Lemma~\ref{lem:MDP-as-uniform} there is an MD-strategy~$\hat\zstrat$ such that for all states $s \in \states$ we have $\probm_{\mdp,\state,\hat\zstrat}(\formula')=1$.
That is, $\hat\zstrat$ reaches the set $\colorset \states = 2$ with probability~$1$, regardless of the start state.
It follows that it reaches, with probability~$1$, the set $\colorset \states = 2$ infinitely often.
Hence $\probm_{\mdp,\state,\hat\zstrat}(\formula)=1$ holds for all $s \in \states$.
\end{proof}

%\newcommand{\thmblabla}{
%Let
%}
%

%\subsection{Proofs of Section~\ref{sec-finite-branch}}

\subsection{Proofs of Subsection~\ref{subsec-safety}}
\begin{qproposition}{\ref{prop:optimal-avoiding}}
\propoptimalavoiding
\end{qproposition}

\begin{proof}
We can assume that $\reachset$ is a sink.
Fix a state $s_0$.
Write $s_0 s_1 s_2 \cdots \in s S^\omega$ for a random run, i.e., $s_1, s_2, \ldots$ denote random states.
For any $n \in \nat$ let $[\next^n \neg \reachset] : s_0 \states^\omega \to \{0,1\}$ be the random variable that indicates if $s_n \not\in \reachset$.
Note that $[\next^n \neg \reachset] \ge \valueof{\mdp}{s_n}$.
Writing $\mathcal{E}_{\mdp,s_0,\optav}$ for the expectation with respect to~$\probm_{\mdp,s_0,\optav}$, we have:
\begin{align*}
& \probm_{\mdp,s_0,\optav}(\formula) \\
& = \probm_{\mdp,s_0,\optav}\Big(\bigcap_{i=0}^\infty \denotationof{\next^i \neg \reachset}{s_0}\Big)
 && \text{semantics of~$\formula$} \\
& = \lim_{n\to\infty} \probm_{\mdp,s_0,\optav}\Big(\bigcap_{i=0}^n \denotationof{\next^i \neg \reachset}{s_0}\Big)
 && \text{measures are continuous from above} \\
& = \lim_{n\to\infty} \probm_{\mdp,s_0,\optav}(\next^n \neg \reachset)
 && \text{$\reachset$ is a sink} \\
& = \lim_{n\to\infty} \mathcal{E}_{\mdp,s_0,\optav}([\next^n \neg \reachset])
 && \text{definition of $[\next^n \neg \reachset]$} \\
& \ge \liminf_{n\to\infty} \mathcal{E}_{\mdp,s_0,\optav}(\valueof{\mdp}{s_n})
 && \text{$[\next^n \neg \reachset] \ge \valueof{\mdp}{s_n}$}
\end{align*}
A straightforward induction on~$n$, using the definition of~$\optav$, shows that $\mathcal{E}_{\mdp,s_0,\optav}(\valueof{\mdp}{s_n}) = \valueof{\mdp}{s_0}$ holds for all $n \in \nat$.
Hence $\probm_{\mdp,s_0,\optav}(\formula) \ge \valueof{\mdp}{s_0}$.
The converse inequality holds by the definition of the value.
Hence $\probm_{\mdp,s_0,\optav}(\formula) = \valueof{\mdp}{s_0}$.
\end{proof} 

%\begin{qlemma}{\ref{lem:safe-two-levels}}
%\lemsafetwolevels
%\end{qlemma}

%\begin{proof}
%We compute probabilities conditioned under the event $\always \safe{\tau_1}$.
%Since $\safe{\tau_1} \subseteq \colorset \states = 0$, we have $\probm_{\mdp,\state,\optav}(\always \colorset \states = 0 \mid \always \safe{\tau_1}) = 1$.
%From the definition of $\safe{\tau_1}$ and the Markov property we get $\probm_{\mdp,\state,\optav}(\always \colorset \states = 0 \mid \neg \always \safe{\tau_1}) \le \tau_1$.
%Applying the law of total probability and writing $x$ for $\probm_{\mdp,\state,\optav}(\always \safe{\tau_1})$ we obtain:
%\begin{align*}
%\tau_2 
%& \ \le \ \probm_{\mdp,\state,\optav}(\always \colorset \states = 0) && \text{Def.~\ref{def:safeset}} \\
%& \ = \ \probm_{\mdp,\state,\optav}(\always \colorset \states = 0 \mid \always \safe{\tau_1}) \cdot x  \\
%& \hspace{4mm} \mbox{} + 
      %\probm_{\mdp,\state,\optav}(\always \colorset \states = 0 \mid \neg \always \safe{\tau_1}) \cdot (1-x) \\
%& \ \le \ x + \tau_1 \cdot (1-x)
%\end{align*}
%It follows $x \ge \frac{\tau_2 - \tau_1}{1 - \tau_1}$.
%\end{proof}

\begin{qlemma}{\ref{lem:as012-return-to-safe}}
\lemaszotreturntosafe
\end{qlemma}
\begin{proof}
For any $n \in \nat$ define $Z_n = \left(\colorset \states = 0\right)^n$.
That is, $Z_n S^\omega$ is the event that the first $n$ visited states have color~$0$.
For any state $\state \not\in \safe{\tau}$, let $n(\state) \in \nat$ be the smallest number such that $\probm_{\mdp,\state,\zstrat}(Z_{n(s)} S^\omega) \le (1+\tau)/2$.
This is well-defined.

Let $L \subseteq \states^*$ be the set of finite sequences $s_0 s_1 \cdots s_{n-1}$ such that $s_0 \not\in \safe{\tau}$ and $n = n(s_0)$ and $\forall i < n .\, s_i \in \colorset \states = 0 \setminus \safe{\tau}$.
We show for all $s \in \states \setminus \safe{\tau}$ and all $k \in \nat$ that $\probm_{\mdp,s,\zstrat}(L^k \states^\omega) \le \left(\frac{1+\tau}{2}\right)^k$.
We proceed by induction on~$k$.
The case $k=0$ is trivial.
For the induction step, let $k \ge 0$.
We have:
\begin{align*}
& \probm_{\mdp,s,\zstrat}(L^{k+1} \states^\omega) \\
& \le \probm_{\mdp,s,\zstrat}(Z_{n(s)} L^{k} \states^\omega) && \text{as $L \cap \{s\} S^* \subseteq Z_{n(s)}$}\\
& \le \probm_{\mdp,s,\zstrat}(Z_{n(s)} \states^\omega) \cdot \sup_{s' \in \states \setminus \safe{\tau}} \probm_{\mdp,s',\zstrat}(L^{k} \states^\omega) \\
& \le \probm_{\mdp,s,\zstrat}(Z_{n(s)} \states^\omega) \cdot \left(\frac{1+\tau}{2}\right)^k  && \text{induction hypothesis}\\
& \le \left(\frac{1+\tau}{2}\right)^{k+1}  && \text{definition of $n(s)$}
\end{align*}
This completes the induction proof.
Write $\formula := \always \neg\safe{\tau} \land \always \colorset \states = 0$.
We have for all $s \in \states$:
\begin{align*}
\probm_{\mdp,s,\zstrat}(\formula)
& = \probm_{\mdp,s,\zstrat}(L^\omega) && \text{$\denotationof{\formula}{} = L^\omega$}\\
& = \lim_{k \to \infty} \probm_{\mdp,s,\zstrat}(L^k \states^\omega) \\
& \le \lim_{k \to \infty} \left(\frac{1+\tau}{2}\right)^k && \text{as shown above} \\
& = 0 && \text{$\tau < 1$}
\end{align*}
It follows:
\[
\probm_{\mdp,s,\zstrat}(\next^j \formula) = 0 
\qquad \text{for all $s \in \states$ and all $j \in \nat$}
\]
Thus we have:
\begin{align*}
\probm_{\mdp,\state,\zstrat}(\eventually \always \neg\safe{\tau} \land \eventually \always \colorset \states = 0)
& \ = \ \probm_{\mdp,\state,\zstrat}(\eventually \formula) \\
& \ = \ \probm_{\mdp,\state,\zstrat}\Big(\bigcup_{j \in \nat} \denotationof{\next^j \formula}{s}\Big) \\
& \ \le \ \sum_{j \in \nat} \probm_{\mdp,\state,\zstrat}(\next^j \formula) \\
& \ = \ 0
\end{align*}
\end{proof}

\subsection{Proofs of Subsection~\ref{subsec-as012}}

\begin{qlemma} {\ref{lem-as-partition-scheme}}
\lemaspartitionscheme
\end{qlemma}
\begin{proof}
We have:
\begin{align*}
\probm(E)
& = \probm\left( \bigcup_{i \in I} (\playset_i \cap E) \right)
= \sum_{i \in I} \probm( \playset_i \cap E )
= \sum_{i \in I} \probm( \playset_i )
= \probm\left( \bigcup_{i \in I} \playset_i \right) \\
& = \probm(\Omega) = 1 
\end{align*}
\end{proof}

%\subsection{Proofs of Subsection~\ref{subsec-cobuchi}}
%\input{app-coBuchi-finite-branch}
\end{document}